\documentclass[11pt,a4paper]{article}
\usepackage[utf8]{inputenc}
\usepackage[T1]{fontenc}
\usepackage{lmodern}
\usepackage{xspace}
\usepackage{fullpage}
\usepackage{amsfonts,amssymb,amsmath,amsthm,dsfont,bm}
\usepackage{graphicx}
\usepackage{hyperref}
\usepackage[capitalise,nameinlink]{cleveref}
\usepackage[small,labelfont=bf]{caption}  			
\usepackage{subfig}	

\usepackage{filecontents}

\newcommand{\coloneqq}{:=}
\newcommand{\bs}{\boldsymbol}
\newcommand{\bb}{\mathbb}
\newcommand{\cl}{\mathcal}
\newcommand{\ie}{\emph{i.e.},\ }
\newcommand{\eg}{\emph{e.g.},\ }
\newcommand{\im}{{\sf i}}
\newcommand{\rv}{r.v.\@\xspace}
\newcommand{\ts}{\textstyle}

\newcommand{\ud}{{\rm d}}
\newcommand{\iid}{%
  \ifmmode
  \mathrm{i.i.d.}%
  \else%
  i.i.d.\@\xspace%
  \fi%
}
\newcommand{\scp}[2]{\langle #1, #2 \rangle}

\DeclareMathOperator{\disk}{disk}
\DeclareMathOperator{\SNR}{SNR}
\DeclareMathOperator{\WSNR}{WSNR}

\DeclareMathOperator{\Id}{\bs{I}}

\DeclareMathOperator*{\argmin}{arg\,min}

\newcommand{\R}[1]{\bb{R}^{#1}}
\newcommand{\TGV}[1][2]{\text{TGV}_\alpha^{#1}}
\DeclareMathOperator{\diag}{diag}

\usepackage{enumitem}
\setlist[enumerate]{leftmargin=.5in}
\setlist[itemize]{leftmargin=.5in}

\newtheorem{remark}{Remark}

\newtheorem{prop}{Proposition}
\newtheorem{definition}{Definition}
\crefname{hypothesis}{Hyp.}{Hypotheses}
\crefname{prop}{Prop.}{Propositions}
\crefname{def}{Def.}{Definitions}
\crefname{section}{Sec.}{Sections}

\title{Compressive lensless endoscopy with partial speckle scanning}

\author{S. Gu\'erit\thanks{ISPGroup, ELEN and INMA Departments, ICTEAM Institute, UCLouvain, Belgium 
  (\url{st.guerit@gmail.com}, \url{laurent.jacques@uclouvain.be}).},\; S. Sivankutty\thanks{Aix Marseille Univ, CNRS, Centrale Marseille, Institut Fresnel, Marseille, France (\url{siddharth@gmail.com}, \url{herve.rigneault@fresnel.fr}).},\; J. A. Lee\footnotemark[6]\ \thanks{Center of Molecular Imaging, Radiotherapy and Oncology (MIRO), IREC Institute, UCLouvain, Belgium (\url{john.lee@uclouvain.be})},\; H. Rigneault\footnotemark[3]\ \thanks{Lightcore Tehnologies, 37-39 Rue d'Antibes, Cannes, France.},\; and L. Jacques\footnotemark[2]\ \thanks{JAL and LJ are Senior Research Associates with the Belgian F.R.S.-FNRS.}}

\begin{document}

\maketitle

\begin{abstract}
The lensless endoscope (LE) is a promising device to acquire \emph{in vivo} images at a cellular scale. The tiny size of the probe enables a deep exploration of the tissues. Lensless endoscopy with a multicore fiber (MCF) commonly uses a spatial light modulator (SLM) to coherently combine, at the output of the MCF, few hundreds of beamlets into a focus spot. This spot is subsequently scanned across the sample to generate a fluorescent image. We propose here a novel scanning scheme, partial speckle scanning (PSS), inspired by compressive sensing theory, that avoids the use of an SLM to perform fluorescent imaging in LE with reduced acquisition time.
Such a strategy avoids photo-bleaching while keeping high reconstruction quality. We develop our approach on two key properties of the LE: (\emph{i}) the ability to easily generate speckles, and (\emph{ii}) the memory effect in MCF that allows to use fast scan mirrors to shift light patterns.
First, we show that speckles are sub-exponential random fields. Despite their granular structure, an appropriate choice of the reconstruction parameters makes them good candidates to build efficient sensing matrices. Then, we numerically validate our approach and apply it on experimental data. The proposed sensing technique outperforms conventional raster scanning: higher reconstruction quality is achieved with far fewer observations. For a fixed reconstruction quality, our speckle scanning approach is faster than compressive sensing schemes which require to change the speckle pattern for each observation.
\medskip

\end{abstract}

\section{Introduction}
Nowadays, imaging and exploring the human body for clinical purposes is quite common. Every newborn already experienced at least one imaging modality during its fetal development through ultrasound scans. During their life, most people will encounter medical issues and will possibly undergo physical examinations based on imaging techniques.  \emph{In vivo} techniques are traditionally divided into structural and functional imaging. 

Although structural imaging has been around since the early days, functional imaging has been the field of tremendous research over the last decades both in the fields of nuclear medecine (\eg single photon emission computed tomography and positron emission tomography) and optical microscopy. Photonic approaches have the advantage to be simpler to implement. They avoid radioactive labels and can achieve better spatial resolution (sub-micron). However, they usually provide limited field-of-views (FOV) and limited penetration depth due to tissues scattering and absorption as compared to non-optical techniques \cite{Ntziachristos2010,Andresen2016}. To circumvent this latter limitation, optical endoscopy and endomicroscopy are being developed with the goal to provide both structural and functional information.

\label{sec:LE}
\subsection{Lensless endoscopes as ultrathin devices}
\label{sec:le_principle}

For instance, endomicroscopes, already used in clinical applications, use a fiber bundle as waveguide \cite{Thompson2011,Perperidis2020}. Recent developments in adaptive optics, optical fibers and computational imaging techniques have opened a new class of imaging systems called \emph{lensless endoscopes} (LE)~\cite{Cizmar2012,Andresen2016,Boominathan2016,Psaltis2016}. In these implementations, a \emph{single} fiber in combination with wavefront shaping devices is employed as an ultrathin imaging system. The extreme miniaturization of the imaging probe (diameter $\leq 200\,\mu$m)  offers a minimally invasive route to 
image at depths unreachable in microscopy.
In this paper, we focus our work on LE that use multicore fiber (MCF) \cite{Andresen2016} and are made with hundreds of individual single core fibers arranged in a single and monolithic silica waveguide \cite{Andresen2013}. Although the diameter of multicore LE fibers is larger than the diameter of multimode LE fibers \cite{Ploschner2015}, MCF can be made resilient to fiber bending \cite{Tsvirkun2019}. Most importantly, they exhibit a memory effect that allows
the output light pattern to be scanned by simply adding a phase tilt at the MCF entrance \cite{Thompson2011}.

\begin{figure}
  \centering
  \includegraphics{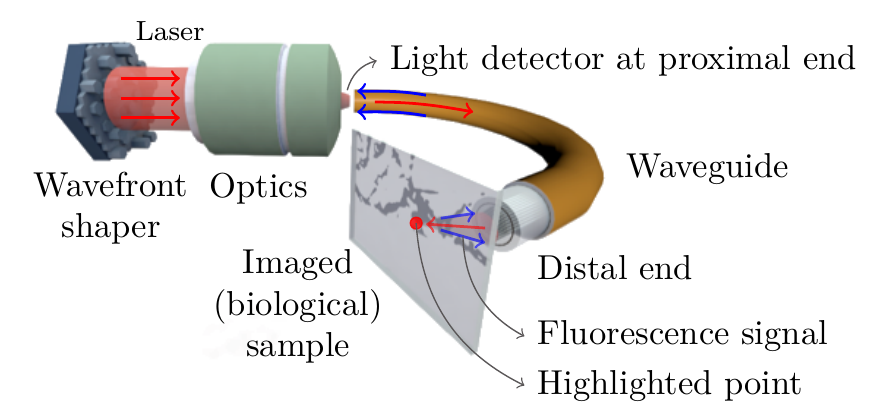}
	\caption{\label{fig:endoscope_principle}Lensless endoscope principle in raster scanning mode (\ie with a focused light beam). Excitation signal is red and backscattered fluorescence signal is blue. Source: Institut Fresnel$^\text{\ddag}$.}
      \end{figure}
      
As illustrated in \cref{fig:endoscope_principle}, an MCF LE consists of four main parts: a wavefront shaper, an optics part, an MCF and an optical detector. The role of the wavefront shaper is to appropriately shape the phase of the light that is injected into the individual cores. This results in the formation of 
specific illumination patterns at the distal end of the MCF. The optics part is made of mirrors and lenses. It is used to focus the light from the wavefront shaper into the individual cores. To collect the generated fluorescence, the MCF features a double cladding that collects and brings back the fluorescent light towards a high sensitive detector \cite{Andresen2013a}.
The sample image is finally reconstructed from the scanning of the focused beam across the sample (by applying a phase tilt at the MCF entrance) and the simultaneous signal collection with the high sensitive detector.
		
This imaging scheme is known as raster scanning (RS) \cite{Psaltis2016,Sivankutty2016}. 
It consists in scanning at constant rate each (discretized) position in the high sensitive FOV. 
For each position, the fluorescence signal is measured by the single pixel detector. The image is then readily reconstructed, pixel by pixel, without any extra post-processing step. As mentioned before, this scanning technique requires preliminary calibration of the SLM to generate a focus beam at the MCF output.
A focused illumination pattern is displayed in \cref{fig:psf_focused}.
This specific focused beam can be obtained when the fiber cores in the MCF are arranged in a Fermat's golden spiral shape. It provides a larger FOV for the LE \cite{Sivankutty2018} than periodically arranged cores.

\begin{figure}
	\centering
	    \subfloat[\label{fig:psf_focused}]{\includegraphics[height=0.23\linewidth]{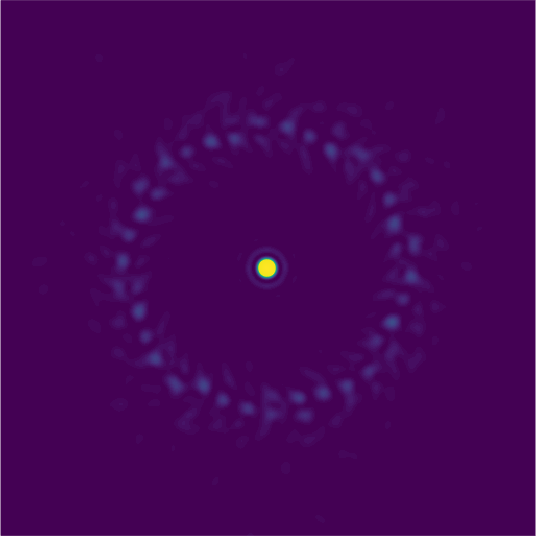}} \hfill
		\subfloat[\label{fig:psf_speckle}]{\includegraphics[height=0.23\linewidth]{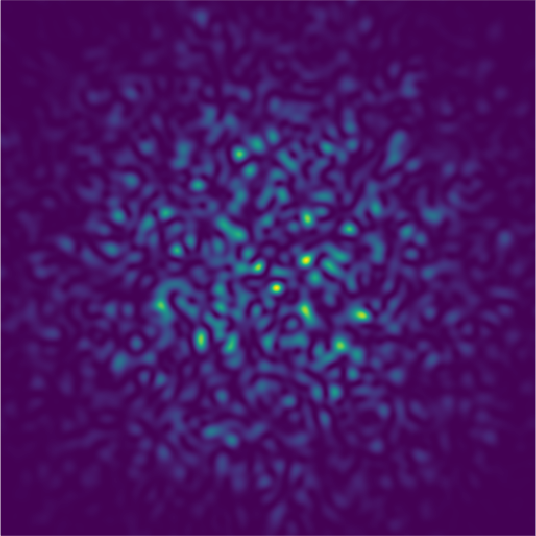}}\hfill
		\subfloat[\label{fig:S0}]{\includegraphics[height=0.23\linewidth]{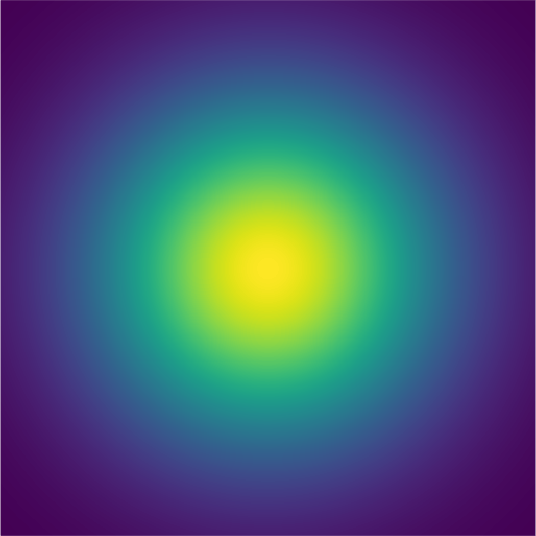}} \hfill
		\subfloat[\label{fig:tildeR}]{\includegraphics[height=0.23\linewidth]{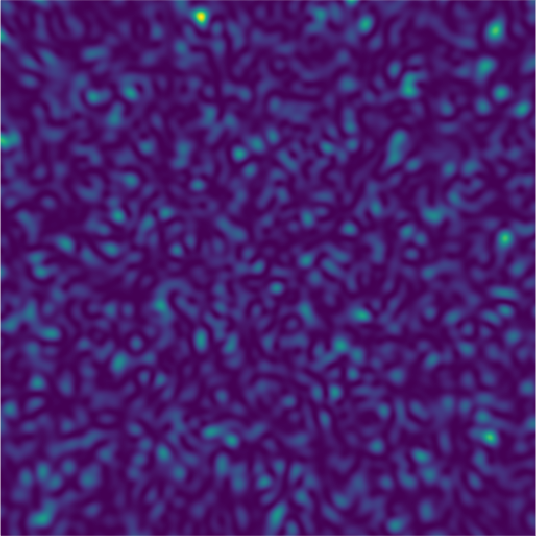}}
	\caption{(a) Focused light beam (brightened for visualization purpose) obtained with the LE with Fermat's golden spiral arranged MCF, (b) speckle pattern (realization of $S$), (c) mean speckle field $\bar S$, and (d) residual field (realization of $\tilde R$). Images were simulated with the following parameters (see \cref{sec:speckles}): $\lambda=1\,\mu$m, $z=500\,\mu$m, $3\sigma_c = 3.2\,\mu$m and $\varpi = 0.5\,\mu$m.}
	\label{fig:S0_full}
\end{figure}

	\subsection{Current challenges}
	\label{sec:current-challenges}
	
	RS acquisition provides fast image reconstruction but at the cost of some drawbacks, for instance: (\emph{i}) the SLM needs to be calibrated to get a focused beam, (\emph{ii}) imperfections of the focused beam are not corrected, and (\emph{iii}) it is necessary to collect as many observations as the number of pixels in the discrete representation of the image. 
	These considerations highlight constraints on the usage of LE in realistic situations: (\emph{i}) the device is sensitive to perturbations during the calibration and imaging, and (\emph{ii}) the slow update rates of conventional SLMs limit the imaging speed \cite{Andresen2013}. 
		
	Hence, the core question of this paper is the following: can we achieve accurate image reconstruction compared to RS using less measurements while keeping a short acquisition time? Without SLM calibration, the observed illumination pattern is a \emph{speckle}, a pattern with multiple bright grains of light and dark regions with no well-defined global peak (see \cref{fig:psf_speckle}). Such a pattern results from the interference of the multiple coherent light beams produced by each single mode cores when they are configured with arbitrary phases.
	In some applications, the presence of a speckle pattern is considered as an issue and methods were developed to suppress it or at least to minimize its effects \cite{Goodman2007}. But in some other imaging modalities, its random structure is appealing as it makes LE closer to recent compressive sensing and computational imaging procedures, \eg indirectly observing images with random sensing strategies~\cite{Candes2008,Duarte2008,Boominathan2016,Psaltis2016}.
	
	Back to 1996, Bolshtyansky \emph{et al.} studied an acquisition scheme very similar to ours \cite{Bolshtyansky1996} but involving a single multimode fiber (MMF). Pseudothermal ghost imaging also relies on a similar approach but uses a rotating diffuser to generate speckle patterns \cite{Katz2009}. Speckles have also been extensively used to perform super-resolution fluorescence microscopy, most often in a blind context  \cite{Mudry2012,Min2013,Kim2015}. 
	To compensate the fact that the patterns may be unknown, some works exploit the statistical properties of the speckles \cite{Mudry2012,Kim2015} or consider weak assumptions on the original fluorophores distribution \cite{Min2013}. Bertolotti \emph{et al.} studied a blind framework where the diffuser is completely opaque and prevents access to both the sample and the generated speckle \cite{Bertolotti2012}. To cover up this inaccessibility, they exploit the \emph{memory effect} summarized as follows: a small change in the incident angle of the light leads to a translation of the speckle pattern. In speckle scanning microscopy, Stasio \emph{et al.} also took advantage of this effect with an MCF fiber \cite{Stasio2016}.	
	Except in super-resolution, the fluorescence signal in the aforementioned applications is measured with a single-pixel detector (or a similar device). The image of interest is estimated via a weighted sum of the speckle patterns \cite{Bolshtyansky1996}, via the correlations of the observation vector with the illumination patterns \cite{Katz2009,Shapiro2012} or via the autocorrelation of the measurement vector \cite{Bertolotti2012,Stasio2016}. 	
	
	While the above-mentioned works consider a number of measurements at least equal to the number of pixels in the final image, several authors reported the development of compressive optical imaging techniques requiring far fewer observations. 
	These techniques are based on the seminal work of Cand\`es \emph{et al}. and Donoho on compressive sensing (CS) \cite{Candes2008,Donoho2006}. 
	The pseudo-randomness of the speckles and their easy generation make them interesting to build an efficient sensing matrix. Katz \emph{et al.} proposed a proof of concept for CS approach to pseudothermal ghost imaging including information about the image structure in the reconstruction algorithm \cite{Katz2009}. 
	A compressive approach was used in super-resolution microscopy \cite{Pascucci2017} where the authors reach higher resolution by using saturated illumination with speckle patterns. 	
	CS was also exploited in experimental setups involving an MCF and/or an MMF, \eg in fluorescence microscopy \cite{Shin2017,Caravaca-Aguirre2018a}, optical photoacoustic \cite{Caravaca-Aguirre2018a} or microendoscopy \cite{Choudhury2019}. Authors assume small total variation \cite{Shin2017} or sparsity in some wavelet basis \cite{Choudhury2019} as prior information on the image structure. 	
	In these three works, speckles patterns are recorded. This recording can be seen as a calibration step. However, as mentioned by \cite{Caravaca-Aguirre2018a}, this step is simpler than the conventional calibration via wavefront shaping because it only requires to measure the speckles intensities while beamforming requires speckle field measurements before the experiment.
	
	\subsection{Proposed approach and contributions}
	\label{sec:proposed_approach}
	
	The literature covered in the previous section shows that speckle illumination seems to be a promising way to face one of the main constraints preventing the use of LE: it removes the need for interferometric stability and costly calibration step before each acquisition, and it can possibly be used in a blind context. Experimental studies show that sensing matrices based on speckles patterns are valid candidates to acquire measurements in a compressive framework while still providing good reconstruction results.
	In fluorescence imaging, more than a low number of measurements, a short acquisition time is highly desirable to preserve the biological sample of photo-bleaching. Photo-bleaching is an irreversible damage resulting from light exposure and leading to a complete loss of fluorescence.
	If the MCF fiber has a low coupling between the cores, this time can be substantially reduced by exploiting the memory effect of the fiber thanks to scan mirrors (\eg galvanometric mirrors) \cite{Andresen2013a}. According to Andresen \emph{et al.}, only strategies exploiting this effect will be able to satisfy the speed and resolution requirements for \emph{in vivo} imaging with LE, as moving such mirrors is much faster than changing an SLM configuration \cite{Andresen2016}. 
	
        In this paper, we study a novel acquisition strategy named \emph{partial speckle scanning} (PSS). It combines elements of compressive sensing with the specific properties of the considered LE, namely its \emph{robustness to spatial and temporal distortion}, its extremely \emph{low coupling} between the cores of the fiber, and its \emph{memory effect}. In particular, we generate only $P \leq M$ distinct speckles 
        and we shift each of them $M_P = M/P \leq M$ times, thus achieving a \emph{partial} scanning of the object, to acquire a total of $M=P M_P$ observations. For $P=M$, we observe the sample with $M$ different speckle illuminations, while if $P=1$, a unique speckle pattern is shifted $M$ times.

	The first contribution of this paper consists in characterizing an accurate linear sensing model of PSS. In particular, we deduce its equivalent sensing matrix, explaining the relation between the fluorescent sample and the collected observations. We then justify the possibility to estimate an image from PSS observations by establishing links with CS theory. This requires us to first analyze the non-asymptotical distribution of a speckle random field as well as its autocorrelation. The second contribution consists in optimizing the sampling parameters ($M_P$, $P$, and the induced speckle shift between two consecutive scanning mirror orientations in a line scanning mode) to get the best reconstruction quality for a given acquisition time or a prescribed number of measurements $M$.

        \subsection{Outline}
        \label{sec:paper-structure}
        The rest of the paper is structured as follows. In \cref{sec:speckles}, we first determine a model for the illumination produced by the MCF in the object plane. We provide a far-field approximation of this illumination pattern in \cref{sec:principle-approx} in function of the electromagnetic field emitted by each core. We explain in \cref{sec:illum-transl} how we can shift this pattern by using scan mirrors. The formation of a focused illumination (used in RS, see~\cref{fig:endoscope_principle}) is explained in \cref{sec:focused-illum} before to introduce and (statistically) characterize in \cref{sec:speckl-illum,sec:speckle-distrib} the speckle patterns produced by randomly setting the core fields. In \cref{sec:model}, we focus on fluorescent imaging with LE. We detail how we can, in general, image a sample by recording the light emitted under illumination. We first develop a continuous forward model in \cref{sec:fluo-model,sec:acquisition_model} relating the illuminated sample---\ie the original fluorophore density---to the recorded observations. In \cref{sec:forward_model_discrete_rep}, we detail the discretization of this model, a mandatory step for any computational imaging method. In \cref{sec:estimation_f}, we then show how this discrete forward model is classically inverted in RS imaging, before to describe in \cref{sec:inverse_problem} a general image estimation algorithm. It inverts the (possibly ill-conditioned) discrete forward model by regularizing the produced image estimate (\eg using total variation or more advanced priors). This algorithm can be solved with iterative proximal methods (\cref{sec:algo}), and we describe in \cref{sec:choice_parameters} a cross-validation strategy to automatically balance the fidelity to the observations with this regularization. In \cref{sec:CS_strategies}, we develop two compressive imaging strategies for LE, namely speckle imaging (SI) in \cref{sec:UCS}, where a distinct random speckle pattern is generated for each observation, and partial speckle scanning (PSS) in \cref{sec:SCS}, where the sample is observed from partial scanning of a few randomly generated speckles. These strategies, their respective formulations as inverse problems and the results of the simulations are also presented in \cref{sec:UCS,sec:SCS}, respectively. Finally, \cref{sec:experiments} presents the experimental setup designed to compare different acquisition strategies on real fluorescent samples. We compare our compressive approaches, SI and PSS, to RS and explain how PSS allows us to reduce the acquisition time compared to SI while providing a similar reconstruction quality. Equivalently, we show that PSS permits to image the sample under limited acquisition time budget when SI fails. We conclude in \cref{sec:perspectives} and provide there possible future improvements for compressive LE imaging.

        \subsection{Conventions and notations} We find useful to introduce here the conventions and notations used throughout this work. Light symbols are used for scalars and functions, while bold symbols are used for vectors and matrices. The ``dot'' notation $f(\cdot; \eta)$ refers to the variability of the function $f$ relatively to the pointed parameter with other parameters (in $\eta$) fixed. The uniform distribution on a set $\cl S$ (\eg $\cl S = [a,b] \subset \bb R$ or the unit ball $\cl S = \bb B_2 \subset \bb R^2$) is $\cl U(\cl S)$. $\cl P(\mu)$ and $\cl N(\mu, \sigma^2)$ are the Poisson distribution with mean $\mu>0$ and Gaussian distribution with mean $\mu$ and variance $\sigma^2$, respectively. We also use the index set $[J] \coloneqq \{1, \cdots, J\}$; the Kronecker symbol $\delta_{j,k}$; the vectors of ones and zeros $\bs 1_d$ and $\bs 0_d$, respectively (the subscript $d$ is omitted when clear from the context); the scalar product $\scp{\bs a}{\bs b} = \bs a^\top \bs b$ between two vectors $\bs a$ and $\bs b$; the $\ell_2$-norm $\|\bs a\| = \sqrt{\scp{\bs a}{\bs a}}$ of a vector $\bs a$; the relation $\bs a \succeq \bs 0$ as a shorthand for $a_i \geq 0$ for all components $i$; the identity matrix $\Id_d$ in $\bb R^d$; Frobenius norm $\|\bs M\|_F$ and spectral norm $\|\bs M\|$ of a matrix $\bs M$; the function $\disk(\bs x)$ equal to 1 if $\|\bs x\| \leq 1$ and 0 otherwise; the symbols ``$^*$'' and ``$^\star$'' for complex conjugation and Legendre-Fenchel conjugate, respectively; the notation $A \lesssim B$ for meaning that $A \le c B$ for some constant $c>0$ independent of $A$ and $B$; and finally, the Fourier transform of $f$ defined as $\hat f(\bs k) \coloneqq \cl F[f](\bs k) = \int_{\bb R^2} f(\bs x)\, e^{-\im \bs k^\top \bs x}\,\ud \bs x$, with inverse $\cl F^{-1}[g](\bs x) = \frac{1}{(2\pi)^2} \int_{\bb R^2} g(\bs k) e^{\im \bs k^\top \bs x}\,\ud \bs k$.

\section{Illumination modeling}
\label{sec:speckles}
Before going into the details of the acquisition modeling and subsequent sensing strategies, let us have a close look at the models and principles explaining the illumination produced by the considered MCF. We first study how we can set the electromagnetic field of each MCF core to focus the illumination on a small spot on the object plane, before to consider the \emph{speckle illumination} formed by randomly configuring the core field. We also provide a statistical analysis of speckle illumination by studying the first- and second-order statistics of this random pattern as well as an characterization of its distribution. Our conclusions will serve the sensing models developped in \cref{sec:model}.  

\subsection{Principles and approximations}
\label{sec:principle-approx}

Given a laser beam of wavelength $\lambda$ injected into a MCF with $J$ cores (as illustrated in \cref{fig:endoscope_principle}), we consider an illumination pattern $S$ in a plane $\cl Z$, parallel to the planar MCF endface $\cl Z_0$, and at a distance $z>0$ from it.

This pattern results from the interferences of the electromagnetic fields radiated by the $J$ cores, each centered on a location $\bs q_j \in \bb R^2$ ($j \in [J]$) in $\cl Z_0$. The field of the $j^{\rm th}$ core in fiber endface is well described by a complex amplitude $\alpha_j$ with $|\alpha_j|=1$ multiplied by a narrow Gaussian envelope~\cite{Sivankutty2016}. These amplitudes are collectively represented by the vector $\bs \alpha \coloneqq (\alpha_1,\,\cdots, \alpha_J)^\top$.

In this context, the pattern $S$ is the intensity of the electromagnetic field $E_z$ radiated from $\cl Z_0$ to the plane $\cl Z$. Fourier optics tells us that, at each location $\bs x \in \cl Z \subset \bb R^{2}$, $E_z$ is obtained via the angular spectrum representation\footnote{Assuming optical field in homogeneous, isotropic, linear and source-free medium \cite{Goodman2005}.}~\cite{Goodman2005}:
\begin{equation}
  \label{eq:angspecrepr} 
  \ts E_z(\bs x; \bs \alpha) = (E_0(\cdot; \bs \alpha) * H_z)(\bs x), 
\end{equation}
where $E_0(\cdot; \bs \alpha)$ is the electromagnetic field emanating from the MCF endface,
\begin{equation}
  \ts E_0(\bs x; \bs \alpha) = \left[U_0 * \sum_{j=1}^J \alpha_j \delta(\bs\cdot - \bs q_j)\right](\bs x),
  \label{eq:E0}
\end{equation}
and $H_z$ is the inverse Fourier transform of the angular spectrum propagator defined as $\hat{H}_z(\bs k) = \exp(- \im z (k^2 - \|\bs k\|^2)^{1/2})$ with $k=2\pi/\lambda$ \cite{Goodman2005}. The vector $\bs k = (k_x, k_y)^\top$ is the 2-D coordinate in the reciprocal space of $\bs x$ and $U_0$ is the field emitted by a single core. The field $U_0$ is approximately Gaussian, \ie $U_0(\bs x) \approx \exp(- \|\bs x\|^2/(2 \sigma_{\rm c}^2))$, with standard deviation $\sigma_c$ depending on wavelength $\lambda$ and related experimentally to mode-field diameter $d$ through $d=2.35 \sigma_{\rm c}$ \cite{Sivankutty2016}. The illumination pattern on a plane located at a distance $z$ of the fiber distal end is therefore
\begin{equation}
  \label{eq:Sx}
  \ts S(\bs x; \bs \alpha) \coloneqq |E_z(\bs x; \bs \alpha)|^2 = \big|\sum_{j=1}^J \alpha_j U_z(\bs x - \bs q_j)\big|^2,
\end{equation}
with $U_z(\bs x) \coloneqq (U_0 * H_z)(\bs x)$.   This last expression can be simplified with the \emph{far-field} assumption. This assumption consider that the illumination pattern is far from the MCF endface, \ie $z \gg k D^2$ with $D$ the diameter of the MCF distal end. In this case, the Fraunhofer approximation holds \cite{Goodman2005}. Combined with the paraxial approximation which supposes a narrow FOV compared to $z$, \ie $\|\bs x\| \ll z$, we express $E_z$ as a modulated scaling of the Fourier transform $\hat E_0$ (computed relatively to $\bs x$) \cite{DiMarzio2020},
\begin{equation}
  E_z(\bs x; \bs \alpha) \approx \ts -\frac{e^{-\im k z}}{\im\lambda z}e^{-\frac{\im k}{2z} \|\bs x\|^2} \hat E_0\big(-\frac{k}{z}\bs x; \bs \alpha\big).
  \label{eq:fraunhofer}
\end{equation}
Therefore, using \cref{eq:E0} and the convolution theorem, the far-field approximation of the illumination $S$ reads
\begin{equation}
  \ts S(\bs x; \bs \alpha) \approx  \frac{1}{(\lambda z)^2}\big|\hat U_0\big(-\frac{2\pi}{\lambda z}\,\bs x\big)\big|^2\Big|\sum_{j=1}^J \alpha_j e^{\frac{2\pi\im}{\lambda z}\bs q_j^\top \bs x}\Big|^2.
  \label{eq:fraunhofer_speckle}
\end{equation}

We introduce now a convenient rewriting of \cref{eq:Sx}. We first define the \emph{mean intensity} $\bar S$ corresponding to non-interacting fields, \ie to the sum of all intensity fields produced by each core in $\cl Z$:
\begin{equation}
  \label{eq:Sbar-def}
  \ts \bar S(\bs x) \coloneqq \sum_{j=1}^J |U_z(\bs x - \bs q_j)|^2 = \big[|U_z|^2 * \text{AF}\big](\bs x),
\end{equation}
where $\text{AF}(\bs x)\coloneqq \sum_{j=1}^J \delta(\bs x  - \bs q_j)$ is the array factor (AF) related to the spatial arrangement of the cores~\cite{Sivankutty2016}.
Quantity $|U_z|^2$ only depends on wavelength $\lambda$, the diameter of each core, and distance $z$. \cref{fig:S0} represents $\bar S$ for a Fermat's spiral core arrangement. Regarding \cref{eq:fraunhofer_speckle}, an approximation for the mean intensity is
\begin{equation}
  \ts\bar S(\bs x) \approx  \frac{J}{(\lambda z)^2}\big|\hat U_0\big(-\frac{2\pi}{\lambda z}\, \bs x\big)\big|^2.
  \label{eq:fraunhofer_Sbar}
\end{equation}

From \cref{eq:Sx}, since the illumination $S$ amounts to summing $\bar S$ and all the cross-terms $\sum_{j,k=1; j\neq k}^J \alpha_j\alpha_k^* U_z(\bs x - \bs q_j)U^*_z(\bs x - \bs q_k)$, we can study the variations of $S$ around $\bar S$ by defining the residual field 
  \begin{equation}
    \tilde R(\bs x; \bs \alpha) \coloneqq \bar S(\bs x)^{-1} \big(S(\bs x; \bs \alpha) - \bar S(\bs x)\big),\ \text{such that}\ S(\bs x; \bs \alpha) = \bar S(\bs x) \big(1 + \tilde R(\bs x; \bs \alpha) \big).
    \label{eq:final_Sx}
  \end{equation}
  This field accounts for the (constructive and destructive) interferences between the individual fields emitted by the cores. Under the Fraunhofer approximation, \cref{eq:fraunhofer_Sbar,eq:fraunhofer_speckle} provide
  \begin{equation}
    \label{eq:resid-field-F-app}
    \ts \tilde R(\bs x; \bs \alpha) \approx \tfrac{1}{J} \sum^J_{j,k=1} \alpha_j\alpha_k^* e^{\frac{2\pi\im}{\lambda z}  (\bs q_j - \bs q_k)^\top \bs x} - 1.
  \end{equation}
As expressed in \cref{eq:final_Sx}, $\bar S$ acts like a \emph{vignetting} window, an envelope on the intensity variations encoded in $\tilde R$ for any configurations of $\bs \alpha$ (see \cref{fig:S0,fig:tildeR}). Moreover, \cref{eq:resid-field-F-app} shows  that $|\tilde R(\bs x; \bs \alpha)| = \cl O(J)$ and we prove in~\cref{sec:speckle-distrib} that, for random complex amplitudes $\bs \alpha$, \ie when $S$ is a speckle pattern, the probability that $|\tilde R(\bs x; \bs \alpha)|$ strongly deviates from a threshold $t>0$ decays exponentially fast when $t$ increases.

\subsection{Translating the illumination}
\label{sec:illum-transl}

We can translate any illumination pattern produced by the MCF in the object plane $\cl Z$ by leveraging the low coupling between the $J$ MCF cores, \ie the MCF \emph{memory effect} \cite{Andresen2013}. This expresses the fact that, up to an additive constant term depending on intrinsic properties of the core, the phase of the light emitted from $j^{\rm th}$ core is the same as the phase of the input light~\cite{Andresen2016}. Therefore, by modulating the light incident to the MCF---by using galvanometric scan mirrors, as in \cref{fig:endoscope_principle}---we act on the complex amplitude of each core field at the distal end of the fiber, and we can shift the illumination pattern at distance $z$ of the fiber. In particular, we can translate the residual field $\tilde R$ defined in \cref{eq:final_Sx} within the intensity vignetting imposed by the mean speckle field $\bar S$.
  
This effect, which is valid for any configuration of the complex amplitudes, is easily explained as follows. A relative tilt $\bs \theta = (\theta_x,\theta_y)^\top$ of the scan mirrors modifies the light optical path on the different core locations and induces a phase ramp\footnote{We here adopt a normalization of the phase ramp by $\frac{2\pi}{\lambda z}$ that eases our next developments; the actual tilt is $\bs \theta' = \frac{\lambda z}{2\pi} \bs\theta$.} $\frac{2\pi}{\lambda z} \bs\theta^\top \bs x$ on the complex amplitude. The vector $\bs \alpha$ is modified as
  \begin{equation}
    \label{eq:gamma-mod-def}
  \bs \alpha' \coloneqq \diag(\bs \gamma(\bs \theta))\,\bs \alpha,\quad \text{with}\ \gamma_k(\bs \theta) \coloneqq e^{\frac{2\pi \im}{\lambda z} \bs \theta^\top \bs q_k},\ k \in [J].    
  \end{equation}
 Therefore, considering the far-field approximation of the residual field \cref{eq:resid-field-F-app}, which is valid using the Fraunhofer and the paraxial approximations, this amplitude modulation changes the residual field into 
  \begin{align}
    \ts \tilde R(\bs x; \bs \alpha')&\ts \approx \tfrac{1}{J} \sum^J_{j,k=1} \alpha'_j{\alpha'_k}^* e^{\frac{2\pi\im}{\lambda z}  (\bs q_j - \bs q_k)^\top \bs x} - 1\nonumber\\
                               &\ts = \tfrac{1}{J} \sum^J_{j,k=1} \alpha_j\alpha_k^* e^{\frac{2\pi\im}{\lambda z}  (\bs q_j - \bs q_k)^\top (\bs x + \bs \theta)} - 1 \approx \tilde R(\bs x + \bs \theta; \bs \alpha),
                                 \label{eq:tildeR_translation}
  \end{align}
  so that, from \cref{eq:final_Sx}, the speckle field becomes           
  \begin{equation}
    S (\bs x; \bs \alpha') \approx \bar S(\bs x) (1 + \tilde R(\bs x + \bs \theta; \bs \alpha) ).
    \label{eq:Sx_tilde_translation}
  \end{equation}
Thus, provided that the far-field conditions are respected, \ie if $\|\bs x\| \ll z$ and $z \gg D_0^2/\lambda$, a non-zero tilt $\bs \theta$ induces a shift $-\bs \theta$ of the residual field, but the vignetting window remains unchanged. 

\subsection{Focused illumination}
\label{sec:focused-illum}

For some arrangements of the fiber cores, such as the Fermat's spiral configuration~\cite{Sivankutty2016}, we can focus the intensity pattern $S$ on a narrow intensity spot (see \cref{fig:psf_focused}). From a convenient calibration of the optical system~\cite{sivankutty2019non}, this focused beam is obtained by programming the SLM so that each core field has unit complex amplitude in the fiber endface. This induces constructive interferences of the propagated core fields in the origin of the object plane $\cl Z$. Mathematically, this can be seen by adapting \cref{eq:Sx}, or \cref{eq:fraunhofer} in the Fraunhofer approximation, to this particular unit amplitude configuration, which provides 
\begin{align}
  \ts S_{\rm foc}(\bs x) \coloneqq S(\bs x; \bs 1_J) = |\sum_{j=1}^J U_z(\bs x - \bs q_j)|^2 &\ts \approx \frac{1}{(\lambda z)^2}\big|\hat U_0\big(-\frac{2\pi}{\lambda z}\,\bs x\big)\big|^2\,\Big|\sum_{j=1}^J e^{\frac{2\pi\im}{\lambda z}\bs q_j^\top \bs x}\Big|^2\nonumber\\
                                                                                          &\ts = \frac{1}{(\lambda z)^2}\big|\hat U_0\big(-\frac{2\pi}{\lambda z}\,\bs x\big)\big|^2\,\big| \cl F[ {\rm AF}](\frac{2\pi}{\lambda z} \bs x)\big|^2\nonumber\\
                                                                                          &\ts = J^{-1}\, \bar S(\bs x)\ \big| \cl F[ {\rm AF}](\frac{2\pi}{\lambda z} \bs x)\big|^2.
                                                                                            \label{eq:focused-beam-def}
\end{align}
In the far-field, the focused beam $S_{\rm foc}(\bs x)$ is thus also vignetted by $\bar S$; it is given by the energy spectral density of the array factor of the fiber cores, as represented in \cref{fig:psf_focused} for the focused beam achieved by Fermat’s spiral core arrangement~\cite{Sivankutty2016}. 

\subsection{Speckle illumination}
\label{sec:speckl-illum}

Speckle patterns were observed for the first time at the end of the 19$^{\rm th}$ century from astronomical observations~\cite{Exner1905,Goodman2007}. They result from the interferences between coherent light components with random relative delays, as induced by light reflection on a rough surface. In this work, we propose to generate such speckle pattern with the MCF by randomly configuring the complex amplitudes $\bs \alpha$ of the intensity field $S(\bs x, \bs \alpha)$ in \cref{eq:Sx}. We set these amplitudes so that their phases follow a uniform distribution on the complex circle, \ie $\alpha_j \sim_{\iid}\,\exp\big(\im\, \cl U([0, 2\pi])\big)$ for $j \in [J]$. This random configuration can be reached by properly shaping the laser beam thanks to the SLM (see \cref{fig:endoscope_principle}). 

In this random configuration, we easily show that, from \cref{eq:Sx}, the mean intensity $\bar S$ is also the expectation of $S$ since $\bb E \alpha_j \alpha_k^* = \delta_{jk}$, \ie 
\begin{equation}
 \ts \bb E S(\bs x, \bs \alpha) = \bb E \big|\sum_{j=1}^J \alpha_j U_z(\bs x - \bs q_j)\big|^2 = \sum_{j=1}^J \big |U_z(\bs x - \bs q_j)\big|^2 = \bar S(\bs x).
\end{equation}

The residual field $\tilde R$ (defined in \cref{eq:final_Sx}) of such a speckle pattern exhibits a spatial granular structure indicating that, for two points $\bs x, \bs x'$ close to each other, the two \rv{s} $\tilde R(\bs x; \bs \alpha)$ and $\tilde R(\bs x'; \bs \alpha)$ are correlated. Given the diameter $D$ of the MCF fiber, we show below that such a correlation exists as soon as $\|\bs x - \bs x'\| \lesssim \lambda z/D$, which shows that the size of a ``speckle grain'' scales like $\cl O(\lambda z/D)$. 
This is achieved by studying the second-order statistics of the random field $\tilde R$: the spatial autocorrelation defined as 
\begin{equation}
  \ts \Gamma_{\tilde R}(\bs x,\bs x + \bs \tau) \coloneqq \bb E [\tilde R(\bs x; \bs \alpha)\tilde R^*(\bs x + \bs \tau; \bs \alpha)],
  \label{eq:gamma_R}
\end{equation}
with $\bs \tau \coloneqq (\tau_x,\tau_y)^\top$. Using the approximation \cref{eq:resid-field-F-app}, $\bb E[\alpha_j\alpha_k^*] = \delta_{jk}$, and $\bb E[\alpha_j\alpha_k^*\alpha_l^*\alpha_m] = \delta_{jk}\delta_{lm} + \delta_{jl}\delta_{km} - \delta_{jklm}$, we find 
\begin{align}
  \Gamma_{\tilde R}(\bs x,\bs x+\bs \tau) & \approx \ts \frac{1}{J^2} \sum_{j,k,l,m} \bb E[\alpha_j\alpha_k^*\alpha_l^*\alpha_m] e^{\frac{2\pi\im}{\lambda z}(\bs q_j - \bs q_k)^\top \bs x} e^{-\frac{2\pi\im}{\lambda z}(\bs q_l - \bs q_m)^\top (\bs x + \bs\tau)}\ -\ 1 \nonumber\\
                                          & \ts = \frac{1}{J^2}\sum_{j,l=1}^J 1 + \frac{1}{J^2}\sum_{j,k=1}^J e^{-\frac{2\pi\im}{\lambda z}(\bs q_j - \bs q_k)^\top \bs\tau} -\frac{1}{J^2}\sum_{j=1}^J 1\ -\ 1 \nonumber\\
                                          & \ts = \frac{1}{J^2} \big|\sum_{j=1}^J e^{-\frac{2\pi\im}{\lambda z} \bs q_j ^\top \bs\tau}\big |^2 - \frac{1}{J} \nonumber\\
                                          &\ts =  \frac{1}{J^2} \Big|\mathcal F\big[\text{AF}\big]\big(\frac{2\pi}{\lambda z}\bs\tau\big)\Big|^2 - \frac{1}{J} \approx \Gamma_{\tilde R}(\bs 0, \bs \tau) =:~\Gamma_{\tilde R}(\bs \tau). \label{eq:gamma_R_final}
\end{align}
Therefore, the Fraunhofer approximation shows us that the autocorrelation $\Gamma_{\tilde R}(\bs x,\bs x+\bs \tau) \approx \Gamma_{\tilde R}(\bs \tau)$ does not depend on $\bs x$ anymore. We also see that, in the far-field, $\Gamma_{\tilde R}$ is directly related to the Fourier transform of the array factor AF, and it displays a peak on the origin with $\Gamma_{\tilde R}(\bs 0) = 1 - J^{-1} \leq 1$. Note that the Fermat's spiral configuration is designed for ensuring small variations of $\Gamma_{\tilde R}$ far from the origin, thus keeping only a dominant peak in $\bs 0$~\cite{Sivankutty2016}.  

The decay of $\Gamma_{\tilde R}$ when $\tau = \|\bs \tau\|$ increases can be further estimated by assuming the core locations homogeneously distributed over 2-D fiber endface. Following~\cite{Goodman2007}, assuming these locations distributed uniformly at random over 2-D fiber endface of diameter $D$, \ie $\bs q_j \sim_{\iid} \cl U(\frac{1}{2} D \bb B^2)$ with $\bb B^2$ the unit ball in $\bb R^2$, we find, for $J$ large and $\bs q, \bs q' \sim_{\iid} \cl U(\frac{1}{2} D \bb B^2)$,
\begin{align}
  \ts \Gamma_{\tilde R}(\bs \tau)&\ts \approx \frac{1}{J^2} \sum_{j,k=1; j\neq k}^J e^{-\frac{2\pi\im}{\lambda z} (\bs q_j - \bs q_k)^\top \bs\tau}\nonumber\\
  &\ts \approx \frac{J-1}{J} \bb E e^{-\frac{2\pi\im}{\lambda z} (\bs q - \bs q')^\top \bs\tau} = \frac{J-1}{J} \big|2(\tfrac{\pi  D}{\lambda z} \tau)^{-1} J_1(\frac{\pi D}{\lambda z} \tau)\big|^2,
  \label{eq:gammaR_circular}
\end{align}
where $J_1$ is the first-order Bessel function of the first kind, and $\lim_{s \to 0} J_1(s)/s = 1/2$.
Following \cite{Goodman2007}, we define the average area of a speckle grain as 
$\mathcal A \coloneqq \iint_{-\infty}^\infty \big(\Gamma_{\tilde R}(\bs \tau) / \Gamma_{\tilde R}(\bs 0)\big) \ud\bs\tau$. For a circular aperture as considered above, $\mathcal A$ is obtained using Parseval's identity and the average radius is defined as $r \coloneqq \sqrt{\mathcal A/\pi}$, which gives
\begin{equation}
  \ts \mathcal A = \frac{(\lambda z)^2}{\pi(D/2)^2} \quad \text{and} \quad r = \frac{\lambda z}{\pi(D/2)}.
  \label{eq:speckle_radius}
\end{equation}
We thus get $r \propto \lambda z/D$. For the parameters of the simulation in \cref{fig:autocorrelations}, the radius of a speckle grain is $2.8\,\mu$m.

\begin{remark}
Quantities \cref{eq:speckle_radius,eq:gammaR_circular} are easily computable since they only depend on the experiment parameters $J$, $D$, $\lambda$ and $z$. \cref{fig:autocorrelations} shows $\Gamma_{\tilde R}$ as a function of $\bs\tau$ (with $\bs\tau = \tau \bs e_1$ and $\bs e_1$ the unit vector aligned with the horizontal axis of the 2-D image) obtained with Fraunhofer approximation and the stronger assumption of a circular spot. Mean empirical estimations of $\Gamma_{\tilde R}$ based on 1\,000 realizations of the residual field are also displayed for $\bs x = \bs 0$.
  Both theoretical approximations explain well the main peak of the autocorrelation, that is similar to the peak observed when the light beam is focused (see \cref{fig:psf_focused}). Even if distant side lobes are absent with circular aperture assumption, \cref{eq:gammaR_circular} still provides a reliable closed form expression for the radius of a speckle grain. 
\end{remark}

\begin{remark}
Interestingly, \cref{eq:gamma_R_final} shows us that the focused beam \cref{eq:focused-beam-def} is actually related to the autocorrelation $\Gamma_{\tilde R}$ of a random speckle, \ie  $S_{\rm foc}(\bs x)/\bar S(\bs x) \approx J \Gamma_{\tilde R}(\bs x) + 1$ since  $J^2\Gamma_{\tilde R}(\bs \tau) + J \approx |\mathcal F[\text{AF}](\frac{2\pi}{\lambda z}\bs\tau\big)|^2$.
\end{remark}

\begin{figure}
  \centering
  \includegraphics{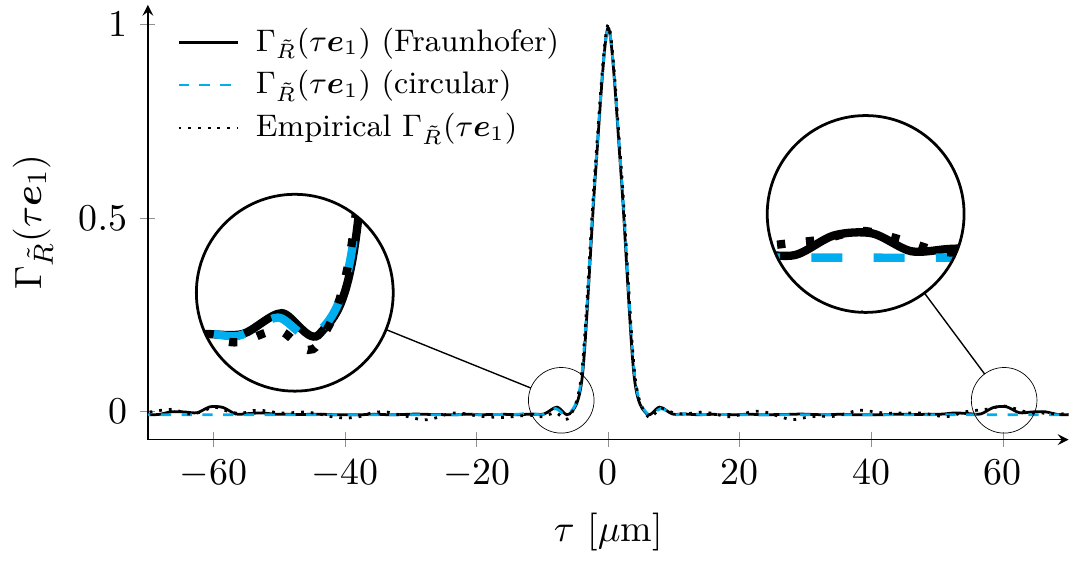}
  \caption{\label{fig:autocorrelations}Autocorrelation of field $\tilde R$ for three cases: (\emph{i}) using Fraunhofer approximation (solid line), (\emph{ii}) using Fraunhofer approximation and assumption of a circular aperture (dashed cyan line), and (\emph{iii}) by empirically computing mean (over 25 trials) autocorrelation based on 1,000 realizations of $\tilde R$ (dotted line). The fiber core arrangement is a Fermat's golden spiral with $J=120$ cores, diameter $D=113\,\mu$m and $3\sigma_c = 3.2\,\mu$m. Experiment parameters are $\lambda=1\,\mu$m, $z=500\,\mu$m, and $\varpi = 2\,\mu$m.}
\end{figure}

\subsection{Speckle distribution}
\label{sec:speckle-distrib}

In \cite{Goodman2007}, Goodman derives the probability density function of a speckle with a large number of phasors (corresponding to the cores in our case). The use of the central limit theorem leads to an intensity $S(\bs x; \bs \alpha)$ distributed according to an exponential law \cite[chapter 3]{Goodman2007}. This is valid in the \emph{asymptotic} case when the number of cores tends to infinity. In this section, we show that the (non-asymptotic) distribution of $S(\bs x; \bs \alpha)$ is sub-exponential when we choose \iid complex amplitudes $\alpha_j$, no matter their distribution. 

Following \cite{Vershynin2010,Vershynin2016a}, a \emph{sub-exponential} (or \emph{sub-Gaussian}) \rv $X \in \bb C$ is characterized by a finite sub-exponential (resp.\,sub-Gaussian) norm $\|X\|_{\psi_1}$ (resp.\,$\|X\|_{\psi_2}$) defined by
  \begin{equation}
    \label{eq:subexp-norm}
    \ts \|X\|_{\psi_p}\ \coloneqq\ \sup_{q\ge 1} q^{-1/p} (\bb E |X|^q)^{1/q},\quad p \in \{1,2\}.  
  \end{equation}
  For instance, Gaussian and Laplacian \rv{s} have finite sub-Gaussian and sub-exponential norms, respectively. For a bounded \rv both norms are finite, \ie  if $|X| \leq C$, then $\|X\|_{\psi_p} \leq C$ for $p \in \{1,2\}$; for our complex amplitudes, $\|\alpha_j\|_{\psi_1} \leq 1$ and $\|\alpha_j\|_{\psi_2}  \leq 1$ for $j \in [J]$.

  These norms are totally equivalent to characterize the tail distribution of their \rv{s}. In fact, for $p \in \{1,2\}$, the norm $K \coloneqq \|X\|_{\psi_p}$ is finite \emph{if and only if} the tail distribution of $X$ satisfies $\bb P \{|X| \geq t \} \leq 2 e^{- c t^p /K^p}$ for all $t\ge 0$, for some universal constant $c>0$. This is similar to the known tail distributions of Gaussian or Laplacian \rv{s} for $p=2$ and $p=1$, respectively.

  \begin{prop} Given a location $\bs x$ and $J$ \iid complex random amplitudes $\{\alpha_j, j \in [J]\}$ with unit modulus, both the speckle intensity $S(\bs x; \bs \alpha)$ and the residual field $\tilde R(\bs x; \bs \alpha)$ are sub-exponential, and 
    \begin{equation}
      \ts	\|S(\bs x; \bs \alpha)\|_{\psi_1} \lesssim \bar S(\bs x)\quad \text{and}\quad \|\tilde R(\bs x; \bs \alpha)\|_{\psi_1} \lesssim 1.
      \label{eq:Sx_tildeR_subexp_norm}
    \end{equation}
    \label{prop:subexp_Sx_tildeR}
  \end{prop}            
  \begin{proof}
    Defining $\bs u(\bs x) \coloneqq \big (U_z(\bs x - \bs q_1), \cdots, U_z(\bs x - \bs q_J)\big)^* \in \bb C^J$, we can rewrite \cref{eq:Sx} as $S(\bs x; \bs \alpha) = |\scp{\bs \alpha}{\bs u(\bs x)}|^2$.
    From \cite[Lemma 5.14]{Vershynin2010}, $\|X\|^2_{\psi_2} \leq \|X^2\|_{\psi_1} \leq 2 \|X\|^2_{\psi_2}$, therefore,  
    $\||\scp{\bs \alpha}{\bs u(\bs x)}|^2\|_{\psi_1} \leq 2 \| \scp{\bs \alpha}{\bs u(\bs x)}\|^2_{\psi_2}$. Moreover, from \cite[Lemma 5.9]{Vershynin2010}, for $J$ \iid sub-Gaussian \rv{s} $X_j$, we have the approximate \emph{rotation invariance} $\|\sum_{j=1}^J X_j\|^2_{\psi_2} \lesssim \sum_{j=1}^J \|X_j\|^2_{\psi_2}$. Note that, for $X_j \coloneqq \alpha_j u^*_j(\bs x)$, $\scp{\bs \alpha}{\bs u(\bs x)} = \sum_{j=1}^J X_j$, using $\|\alpha_j\|_{\psi_2} \leq 1$ and the fact that ${\|\cdot\|_{\psi_2}}$ is a norm, this involves that
    $$
    \ts \| \scp{\bs \alpha}{\bs u(\bs x)}\|^2_{\psi_2} \lesssim  \sum_{j=1}^J \| X_j \|^2_{\psi_2} = \sum_{j=1}^J \| \alpha_j\|^2_{\psi_2} |u_j(\bs x)|^2 \leq \sum_{j=1}^J |u_j(\bs x)|^2 = \bar S(\bs x),
    $$
    which proves the first inequality in \cref{eq:Sx_tildeR_subexp_norm}. Regarding the second inequality, since $\|\cdot\|_{\psi_1}$ is a norm, it respects the triangular inequality, and from \cref{eq:final_Sx}, we get
    \begin{align*}
      \ts \|\tilde R(\bs x; \bs \alpha)\|_{\psi_1}&\ts = \| \bar S(\bs x)^{-1} (S(\bs x; \bs \alpha) - \bar S(\bs x)) \|_{\psi_1}\\
      &\ts \leq \| \bar S(\bs x)^{-1} S(\bs x; \bs \alpha) \|_{\psi_1} + 1 = \bar S(\bs x)^{-1} \|(S(\bs x; \bs \alpha)\|_{\psi_1} + 1 \lesssim 1.
    \end{align*}
 \end{proof}

\cref{prop:subexp_Sx_tildeR} informs us that the tail of the distribution of $\tilde R(\bs x; \bs \alpha)$ decreases exponentially fast with a threshold $t\ge 0$, \ie for some universal constants $C,c>0$, 
  $$
  \bb P \big(\big|\tilde R(\bs x; \bs \alpha)\big| \geq t \big) \le C e^{-c t}.
  $$
  
\begin{figure}
  \centering
  \includegraphics{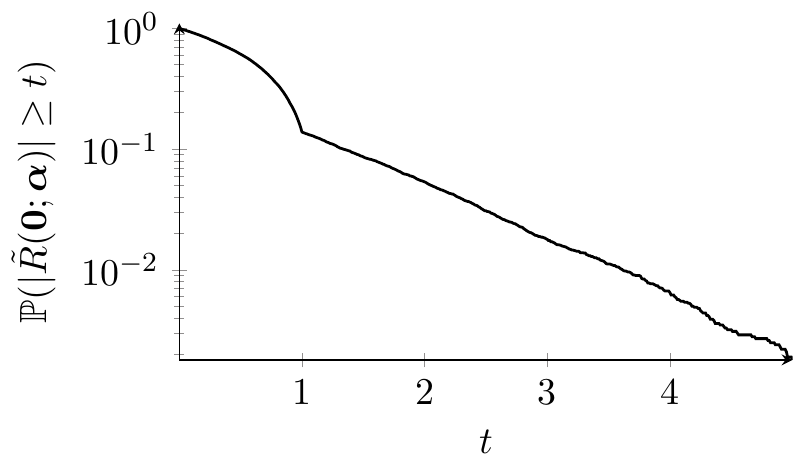}
  \caption{\label{fig:decay_Rtilde}Estimation of $\bb P \big( | \tilde R(\bs 0; \bs \alpha) | \geq t \big)$ as a function of $t\ge 0$ based on the simulation of 10,000 speckle patterns of size $128\times 128$. The fiber core arrangement is a Fermat's golden spiral with $J=120$, $D=113\,\mu$m and $3\sigma_c = 3.2\,\mu$m. Experiment parameters are $\lambda=1\,\mu$m, $z=500\,\mu$m, and $\varpi = 2\,\mu$m. We do observe an exponential decay of the tail of the distribution ($t \ge 1$).}
\end{figure}
\cref{fig:decay_Rtilde} displays an empirical approximation of $\bb P (|\tilde R(\bs 0; \bs \alpha)| \ge t)$ as a function of $t$ based on the generation of 10,000 discretized speckle patterns of size $128\times 128$. For $t \geq 1$ (corresponding to the tail), we do observe an exponential decay of the distribution.

\section{From illumination to fluorescent imaging}
\label{sec:model}
In this section, we first describe the fluorescence phenomenon, \ie the physical process underlying the generation of observations in lensless endoscopy. Second, we describe the model of photons collection by the fiber. Then, based on those previous models, we propose a forward model explaining how the observations are related to the original fluorophore density map. Finally, we formulate the imaging process as the solving of an inverse problem, \ie the minimization of an objective function made of a data fidelity term combined with regularizing priors on the expected density map, and explain how to solve it numerically.

More formally, we consider a sample---either a synthetic phantom or a very thin slice of an \emph{ex vivo} biological sample---containing some molecules of interest that have been tagged with fluorescent dyes. For the sake of simplicity, we assume that this sample is restricted to a 2-D FOV, also called object space $\Omega \subset \mathbb R^2$. A density function $f_0: \Omega \times \mathbb R \rightarrow \mathbb R_+$ associates a fluorophore density $f_0(\bs x,t)$ with each location $\bs x \in \Omega$ at time instant $t$. 
		
	\subsection{Fluorescence model}
	\label{sec:fluo-model}

	Fluorescent molecules possibly absorb light and emit photons under illumination by a light pattern with static intensity $S(\bs x)$. The observed and measured photon flux $\varphi : \Omega \times \mathbb R \rightarrow \mathbb R$ depends on the intensity of the incident light but also on the fluorophore density $f_0$. In this section, we derive the expression of the photon flux $\varphi(\bs x,t)$ at time $t$ emitted by the fluorophores located at $\bs x \in \Omega$ as a function of $f_0$, $S$ and physical parameters. In the next section, we will relate this expression to the measurements obtained with the optical device.
	
	The main phenomena controlling the number of emitted photons are singlet state and triplet state saturations, photo-bleaching, and background fluorescence~\cite{Tsien2006}. We refer the reader to \cite{Song1995,Tsien2006} for detailed explanations about the physics of fluorescence emission.

	In the context of LE, we formulate the following assumptions about the illumination duration and intensity, both tunable acquisition parameters: (\emph{i}) illumination duration $t_{\rm acq}$ is short compared to the time needed to reach triplet excited steady state but long enough for singlet excited state to reach steady state, and (\emph{ii}) intensity $S$ at each $\bs x$ is low and, consequently, far from saturation level. 
	The first assumption allows us to consider only singlet excited state saturation and its corresponding flux. The second assumption amounts to considering that $S(\bs x) \ll k_d/\sigma$, where $k_d$ is the constant rate associated with the return to the ground state and $\sigma$ is a parameter depending on laser wavelength $\lambda$ and on the chosen dye~\cite{Cheng2006,Tsien2006}. 
	In this case, the photon flux emitted at $\bs x$ depends linearly on the product of $S(\bs x)$ and $f_0(\bs x)$,
	\begin{equation}
	\varphi(\bs x) \approx Q_e \sigma\, S(\bs x) f_0(\bs x) .
	\label{eq:flux_linearized}
	\end{equation}
	where quantum yield $Q_e$ characterizes the efficiency of the emission process. This time-invariant model also assumes that there is no molecules displacement during the acquisition---\ie the illumination duration is short---and we neglect photo-bleaching of fluorescent dye.
        
	\subsection{Photon collection model}
	\label{sec:acquisition_model}

	 For a static illumination pattern $S(\bs x)$, the LE collects within the acquisition duration $t_{\rm acq}$ a non-negative number of photons $Y$, a fraction of those emitted by the fluorescent sample. Neglecting for now any other noise sources that could plague this photon collection such as thermal and readout noise in the optical detector, $Y$ is a Poisson \rv whose mean is
	\begin{equation}
	\bb E Y = \int_\Omega \int_0^{t_{\rm acq}} c_0\, \big(\varphi(\bs x)+\varphi_b(\bs x)\big) \ \ud t\,\ud \bs x,
	\label{eq:collection_model}
	\end{equation}
	where $\varphi$ and $\varphi_b$ are the direct and the background photon fluxes generated by the object, and $0<c_0<1$ is a constant accounting for the fraction of photons captured by the LE. As explained in \cref{sec:experimental_setup}, the model \cref{eq:collection_model} assume no sensitivity attenuation of the light collection (as induced by the fiber numerical aperture) closer to the periphery of the FOV.

        	Using \cref{eq:flux_linearized,eq:collection_model}, and assuming $\varphi_b$ is induced by constant density over the FOV, the continuous forward model relating the photon count $Y$ to $f_0$ is equivalent to
	\begin{equation}
	Y \sim \mathcal P\Big(\int_\Omega S(\bs x)[c_1 f_0(\bs x)+c_2]\,\ud \bs x \Big),
	\label{eq:continuous_model_01}
	\end{equation}
	with constants $c_1, c_2 > 0$ depending on $c_0$, $Q_e$, $\sigma$ and $t_{\rm acq}$. In this study, we are interested in the relative contrast between different regions of the object space $\Omega$. Recovering $f_0$ up to a scaling factor and an offset is adequate. Defining $f \coloneqq c_1 f_0 + c_2$, model \cref{eq:continuous_model_01} becomes
	\begin{equation}
	Y \sim \mathcal P\Big(\int_\Omega S(\bs x)f(\bs x)\,\ud \bs x \Big).
	\label{eq:continuous_forward-model}
	\end{equation}

	\subsection{Simplified discrete forward model}
	\label{sec:forward_model_discrete_rep}

	We now simplify \cref{eq:continuous_forward-model} by turning it into a discrete representation. For this, we suppose that both $f$ and $S$ can be represented in a finite set of $N=n^2$ orthonormal functions $\{\psi_i\}_{i=1}^N$ over $\Omega \subset \bb R^2$, \ie there exist coefficients $\{s_i\}_{i=1}^N$ and $\{f_{i}\}_{i=1}^N$ such that $S(\bs x) \approx \sum_{i=1}^N s_i \psi_i(\bs x)$ and $f(\bs x) \approx \sum_{i=1}^N f_{i} \psi_i(\bs x)$. By Parseval, this also means that $\int_{\Omega} S(\bs x) f(\bs x) \ud \bs x \approx \bs s^\top \bs f$, where $\bs s \coloneqq (s_1, \dots, s_N)^\top$ and $\bs f\coloneqq (f_{1}, \dots, f_{N})^\top$. In this work, we select a naive pixel representation with pixel pitch $\varpi>0$, \ie $\psi_i(\bs x) \coloneqq \psi(\bs x - \bs x_i)$ with $\cl X \coloneqq \{\bs x_i\}_{i=1}^N$ evenly sampling $\Omega$ on $N$ locations and $\psi(\bs x)$ equals to $1/\varpi$ if $\bs x \in [- \varpi/2,\varpi/2] \times [-\varpi/2,\varpi/2]$, and 0 otherwise. Assuming a square FOV with side length $L$, we thus set $\varpi = L/n$.

        Combined with the corruption of any possible signal-independent noise sources (\eg electronic and readout noise) appearing in the measurement process, the continuous model \cref{eq:continuous_forward-model} becomes
	\begin{equation}
	Y \sim \mathcal P(\bs s^\top \bs f) +\cl N(0, \sigma^2),
	\label{eq:discrete_model_02}
	\end{equation}
where all extra noise sources are collectively modeled as a Gaussian noise with variance $\sigma^2$.
      
        This work considers an estimation of the biological specimen, assimilated to its fluorophore density map $\bs f$, from the information captured from $M$ distinct observations $\bs y \coloneqq (y_1, \cdots, y_M)^\top \in \bb N^M$, \ie $M$ realizations of $Y$. Each observation $y_j$ is associated with  a discrete illumination patterns $\bs s_j$, the discrete representation of the field $S(\bs x, \bs \alpha_j)$---focused or speckle---set by one complex amplitude configuration $\bs \alpha_j \in \bb C^J$. As justified by our experimental conditions (see \cref{sec:experiments}), this estimation is also realized in a high photon counting regime where the Poisson noise is negligible compared to other noise sources.   

          Gathering the $M$ illumination patterns in the matrix $\bs S \coloneqq (\bs s_1, \cdots, \bs s_M)$ and following the conclusions of \cref{sec:speckles}, we can thus compactly represent the sensing model as 
	\begin{equation}
          \ts \bs y = \bs S^\top \bs f + \bs n, 
          \label{eq:discrete_model_03}
        \end{equation}
      where $\bs n \coloneqq (n_1, \dots, n_M)^\top$ amounts to $M$ realizations of the noise sources $\cl N(0, \sigma^2)$. In practice, while the noise power is generally unknown in \cref{eq:discrete_model_03}, we describe in \cref{sec:choice_parameters} a cross-validation technique that determines the influence of $\bs n$ in the estimation of $\bs f$ from $\bs y$.

	\subsection{Imaging process}
	\label{sec:estimation_f}

        In this section, after a brief review of LE imaging by raster scanning of a focused beam, we present a more general imaging process allowing object observation with speckle illumination. The image reconstruction can then be seen as solving of an inverse problem---accounting for the indirect observation of the fluorophore density map---regularized by some prior information on the structure of this map.

		\subsubsection{Raster scanning imaging}
		\label{sec:raster-scanning}
       
		As mentioned in \cref{sec:le_principle}, RS is a conventional sensing strategy for lensless endoscopy \cite{Sivankutty2016}. It amounts to scanning the object of interest with the focused light beam $S_{\rm foc}(\bs x)$ defined in \cref{eq:focused-beam-def}. As explained in~\cref{sec:illum-transl}, the action of the scanning mirrors allows us to move the center of $S_{\rm foc}$ along a trajectory $\Omega_{\rm tr}$ in the FOV $\Omega$, and to record one observation of the fluorophore density $f$ for each trajectory position. Thus, the number of observations crucially depends on the total length of $\Omega_{\rm tr}$ and its sampling rate, both tunable parameters of the experiment.
		
		Let us assume that we get $M$ observations $\bs y = (y_1, \cdots, y_M)$ for a light beam focused at $M$ discrete locations $\Omega_{\rm d} \coloneqq \{\bs p_i: i \in [M]\} \subset \Omega_{\rm tr}$ that belong to the pixel grid $\cl X$ (see \cref{sec:forward_model_discrete_rep}). In this case, in a noiseless setting and in the discrete model developped in \cref{sec:forward_model_discrete_rep}, RS imaging directly estimates $f$ from the observations, \ie 
	\begin{equation}
	\label{eq:RS-reconstruction}
	\hat f(\bs p_i) = y_i = \bs s_{{\rm foc},i}^\top \, \bs f \approx f(\bs p_i),
      \end{equation}
      where $\bs s_{{\rm foc},i}$ is the discrete representation of $S_{\rm foc}$ when it is translated on $\bs p_i$.
      
	In other words, we consider that the matrix $\bs S$ generated by the patterns $\{\bs s_{{\rm foc},i}\}_{i=1}^M$ in \cref{eq:discrete_model_03} is well approximated by the identity operator over the FOV specified by $\Omega_{\rm d}$. However, even with a fiber arrangement optimized to produce point-like light beams, distant side lobes in the light pattern (see \cref{fig:psf_focused}) limit the accuracy of this approximation.  Moreover, setting unit amplitude $\alpha_j$ in the fiber endface requires to calibrate the SLM, and the FOV must be densely scanned by $\Omega_{\rm tr}$. This last point forces the RS trajectory to visit each point of a fine grid $\cl X$. We explain below how one can alleviate these limitations by adopting a compressive imaging strategy relying on random speckle illumination.

        \subsubsection{Inverse problem formulation}
        \label{sec:inverse_problem}

        In the context of either focused RS imaging or for speckle illumination, we can \emph{invert} the forward model given by \cref{eq:discrete_model_03}. We do so by posing the reconstruction of the fluorophore density map $\bs f$ as the solving of an inverse problem, that is, the estimation of $\bs f$ from the (indirectly) observed data $\bs y$ for a general sensing matrix $\bs S$. A natural but naive way to find this estimate amounts to picking the density $\bs{\hat f}$ minimizing the squared $\ell_2$-norm of the residual $\bs S^\top\bs{\hat f} - \bs y$. This minimization problem is unfortunately often \emph{ill-posed}, mainly due to the noise corruption but also because $\bs S^\top$ is not necessarily well conditioned, \eg if the number of observations $M$ is smaller than the number of pixels $N$. Consequently, we cannot guarantee the uniqueness or even the existence of a solution to this least squares minimization.
	
	We overcome this issue by regularizing the inverse problem, \ie by adding extra objective functions $g_1$ and $g_2$ (or priors) in the minimization \cite{Boyd2004}.
	We solve this multi-criterion problem by minimizing the weighted sum of the objective functions:
	\begin{equation}
		\bs{\hat  f} \in \underset{\bs u}{\text{arg\,min}}~\|\bs S^\top \bs u - \bs y\|^2 + \rho_1 g_1(\bs u) + \rho_2g_2(\bs u),
		\label{eq:dist_regularized_tilde_TGV}	
	\end{equation}
	where $\rho_1,\rho_2 > 0$ are \emph{regularization parameters} balancing between the data fidelity term (and thus the measurement noise level) and regularization terms $g_1$ and $g_2$. We detail in \cref{sec:choice_parameters} how to automatically estimate those parameters. 
        In this work, we focus on two extra priors $g_1$ and $g_2$.
        
	\paragraph{Image structure} We assume that $\bs f$ is made of smooth areas separated by sharp boundaries (possibly corresponding to tissue interfaces or cellular membranes). This prior, encoded in function $g_1$, corresponds either to the minimization of the total variation (TV) norm \cite{Rudin1992} or the second-order total generalized variation ($\text{TGV}_\alpha^2$) norm \cite{Bredies2010,Knoll2011} (see \cref{app:TGV} for formal definitions).  
	Promoting a small TV-norm leads to piecewise constant images and is then well suited for the synthetic data used in the simulations. As for TGV, it can be seen as a generalization of TV to higher-order image derivatives. Minimizing the $\text{TGV}_\alpha^2$-norm leads to piecewise linear estimates where parameter $\alpha > 0$ mades a trade-off between edge-preserving and smoothness-promoting terms (as suggested by \cite{Knoll2011}, we set $\alpha=2$). $\text{TGV}_\alpha^2$ is better suited for fluorescent samples used for the experiments or when we estimate a vignetted map. 
	\paragraph{Non negativity} We know from the physics of fluorescence emission that $\bs f \succeq \bs 0$. Function $g_2$ is then a convex indicator function on the set $\R{N}_+$ that is equal to zero if the constraint is satisfied and to $+\infty$ otherwise, $g_2 = \imath_{\R{N}_+}$.
	
        \subsubsection{Minimization algorithm}
        \label{sec:algo}
        Since we do not require $g_1$ and $g_2$ to be differentiable, we cannot minimize \cref{eq:dist_regularized_tilde_TGV} with methods relying on the computation of the gradient or the Hessian of the objective function. Instead, we resort to the family of proximal algorithms \cite{Parikh2013} that can deal with optimization of non differentiable functions. In a nutshell, these algorithms can minimize the sum of several convex non-smooth functions by splitting this optimization into an iterative algorithm relying on the computation of the \emph{proximal operator} of each of these functions. This operator is defined as follows.
		\begin{definition}(from \cite{Parikh2013})
		Let $\psi$ be a lower semicontinuous convex function from $\mathcal S \subset \bb R^d$ to $\left]-\infty,+\infty\right[$ such that the domain of $\psi$ is non-empty. The proximal operator of $\psi : \mathcal S \rightarrow \mathcal S$ evaluated in $\bs z \in \mathcal S$ is unique and defined as 
		\begin{equation*}  
		\ts \text{\normalfont prox}_{\psi}(\bs z) \coloneqq \underset{\bs x \in \mathcal S}{\text{\normalfont arg\,min}}~\tfrac{1}{2}\|\bs z - \bs x\|^2_2 + \psi(\bs x).
		\label{prox_op}
		\end{equation*}
		\end{definition}
		The evaluation of the proximal operator of a convex function $\psi$ on $\bs z$ thus provides a minimizer of $\psi$ that remains close to $\bs z$. For many smooth and non-smooth convex functions, this operator is closed form or fast to compute; this is the case of the $\ell_1$-norm, the T(G)V-norm, the indicator function of a convex set, and the functions $g_1$ and $g_2$ defined above.
		
		In this work, we solve \cref{eq:dist_regularized_tilde_TGV} with a generalized version of the Chambolle-Pock (CP) primal-dual algorithm \cite{Chambolle2010,Gonzalez2014}. This flexible algorithm allows for composing convex functions with linear operators in the minimization, \ie it can be used to solve the following type of problem~\cite{Gonzalez2014},
		\begin{equation}
                  \ts \underset{\bs u}{\text{min}}~h(\bs u) + \sum_{k=1}^K \varphi_k(\bs A_k\bs u).
                  \label{eq:CP-template}
              \end{equation} 
              Solving \cref{eq:dist_regularized_tilde_TGV} with TV regularization thus corresponds to solving \cref{eq:CP-template} with $K=2$, $h(\bs u)=\imath_{\R{N}_+} (\bs u)$, $\varphi_1(\bs u) = \|\bs u - \bs y\|^2$, $\varphi_2(\bs u) = \|\bs u \|_1$, and, in the case where $g_1$ is the TV-norm, the matrix $\bs A_1 = \bs S^\top$ and $\bs A_2 = \nabla$ is the finite difference operator (defined in \cref{app:TGV}).  If $g_1$ is the $\TGV$-norm, \cref{app:TGV} provides the definition of the functions and the corresponding operators. The CP algorithm only requires computing the proximal operators of $h$, the Legendre-Fenchel conjugate of each $\varphi_k$, as well as the adjoints of the matrices $\bs A_k$'s. Notice that, as an alternative, the Generalized Forward Backward algorithm can also solve \cref{eq:dist_regularized_tilde_TGV} by leveraging the fact that the function $\|\bs S^\top\bs u - \bs y\|^2$ is differentiable with respect to $\bs u$~\cite{Raguet2013}.	 			 
		
              \subsubsection{Regularization parameter and stopping criterion}
              \label{sec:choice_parameters}
              As brought up earlier, the optimal value of the parameter $\rho$ balancing the data fidelity term and the regularization term $g_1$ in \cref{eq:dist_regularized_tilde_TGV} is unknown \emph{a priori}. From a maximum a posteriori standpoint, this parameter actually depends on the chosen regularization term $g_1$ and on noise corrupting the observations $\bs y$ \cref{eq:discrete_model_03}. We can, however, set a value for $\rho$ based on the optimization of a criterion depending on the estimate and the observations.  Boufounos \emph{et al.}, followed by Ward, proposed for instance to use cross-validation (CV) to estimate such a balancing parameter \cite{Boufounos2007,Ward2009}. To avoid overfitting noisy observations in the context of compressive sensing, the set of observations can be divided into two parts: a first set is used for the signal estimation, and a second, much smaller, is used to determine when to stop an iterative reconstruction algorithm, or to select appropriate regularization parameters. 
		
              We propose here to use the same CV approach. We first randomly split the vector of $M$ observations in $\bs y$ in an estimation and validation vectors $\bs y_{\rm est}$ and $\bs y_{\rm val}$ of size $M_{\rm est}$ and $M_{\rm val} = M - M_{\rm est}$, respectively, with $M_{\rm val}$ representing a small fraction of $M$ (see \cref{sec:CS_strategies}). We use for this two restriction matrices, $\bs R_{\rm est} \in \{0,1\}^{M_{\rm est}\times M}$ and $\bs R_{\rm val} \in \{0,1\}^{M_{\rm val}\times M}$, such that $\bs y_{\rm val} = \bs R_{\rm val}\,\bs y$ and similarly for $\bs y_{\rm est}$. In this context, the estimation problem \cref{eq:dist_regularized_tilde_TGV} becomes
              \begin{equation}
	\bs{\hat f}_\rho \in \underset{\bs u}{\text{arg\,min}}~\|\bs R_{\rm est}(\bs S^\top\bs u - \bs y)\|^2 + \rho g_1(\bs u) + \imath_{\R{N}_+} (\bs u),
	\label{eq:dist_regularized_tilde_TGV_CV}
\end{equation}
		As in \cite{Boufounos2007}, we use the squared $\ell_2$-norm of the residual to evaluate the quality of the produced estimate $\bs{\hat f}_\rho$. This is defined by
		\begin{equation*}
		q(\bs{\hat  f}_\rho) \coloneqq \|\bs R_{\rm val}(\bs S^\top \bs{\hat  f}_\rho - \bs y)\|^2.
		\label{eq:criterion_residual}
              \end{equation*}

             This quality function $q$ is then used, first, to stop the CP algorithm when $q$ reaches a minimum, and second, to select an appropriate value for $\rho$. This last operation is done by solving \cref{eq:dist_regularized_tilde_TGV_CV} for $K$ values of $\rho$ and $\rho_{k^\ast}$ is considered as the best choice if $q(\bs{\hat  f}_{\rho_{k^\ast}}) \le q(\bs{\hat f}_{\rho_k})$, for all $k \in [K]$.

\section{Compressive imaging with speckle illumination}
\label{sec:CS_strategies}
We show in this section how the LE acquisition model~\cref{eq:discrete_model_03} can follow a compressive sensing (CS) scheme~\cite{Candes2008}, \ie when the density map $\bs f \in \bb R^N$ is observed with $M < N$ random speckle illuminations. This also aim at removing the beamforming calibration of the RS method by using speckle illumination. We thus consider a sensing model~\cref{eq:discrete_model_03} where the patterns $\{\bs s_k\}_{k=1}^M$ composing $\bs S \in \bb R_+^{N \times M}$ are associated with $M$ random configurations $\{\bs \alpha_k\}_{k=1}^M \subset \bb C^J$ of the MCF complex field amplitudes.

In CS theory, the success of the estimation of a signal of interest mainly depends on the ability of the sensing matrix to capture information about this signal during the acquisition process~\cite{Rauhut2012}. This theory offers tools to study the recovery of low-complexity signals (such as sparse images in a given basis) when the number of observations is much smaller than the number of pixels in the final image, \eg the restricted isometry property (RIP) and the null space property \cite{Candes2008,Rauhut2012}. But as mentioned in \cite{Duarte2011,Rauhut2012,Tanaka2018}, those two properties are difficult to verify directly (\eg for deterministic matrices) without extra assumptions on the sensing matrix. Most of the existing theoretical work in CS theory is about random matrices with zero-mean \iid row entries, or random partial Fourier sampling. For instance, the behavior of a reconstruction based on observations acquired with sub-Gaussian or Bernoulli measurement matrices has been extensively studied from a theoretical point of view \cite{Candes2008,Jacques2010a}. However, those zero-mean random matrices with \iid entries are too ideal for real-life applications where the physics often limits the choice of sensing operator \cite{Duarte2011,Rauhut2012}.

As detailed in the next subsections, we study two practical compressive acquisitions for LE, and we show how to make them, at least partially, compatible with CS theory. The first strategy, called \emph{speckle imaging} (\textbf{SI}), observes the object with $M$ distinct speckle patterns generated by $M$ \iid complex amplitudes configurations $\{\bs \alpha_k\}_{k=1}^M$. The second, coined \emph{partial speckle scanning} (\textbf{PSS}), adopts a structured compressive sensing strategy; while observing the object, only a fraction of the $M$ speckles are generated randomly, and the others consist of shifted copies of the firsts.

\subsection{Speckle imaging}
\label{sec:UCS}

In SI, the phases of the complex amplitudes of the $J$ MCF cores are all independently picked uniformly at random over $[0, 2\pi]$ for each measurement, \ie given $\bs \alpha_k = (\alpha_{k1}, \, \cdots,\alpha_{kJ})^\top$, we have $\alpha_{kj} \sim_\iid \exp(\im\,\cl U([0,2\pi]))$, for $j \in [J]$ and $k \in [M]$. This unstructured compressive sensing schemes thus acquires $M$ observations by successively illuminating the biological sample with $M$ different unstructured light patterns $\bs s_j$, each obtained from a random configuration of the SLM. In this context, we first note that the sensing model~\cref{eq:discrete_model_03} can be rewritten~as 
\begin{equation}
  \ts \bs y = \bs S^\top \bs f + \bs n = \sqrt{M}\, \bs \Phi\, \bs{\bar S} \bs f + \bs n,
  \label{eq:discrete_model_CS}
\end{equation}
where the sensing matrix $\bs \Phi \in \bb R^{M \times N}$ is a renormalization of $\bs S^\top$ defined as
\begin{equation}
  \ts \bs \Phi \coloneqq \tfrac{1}{\sqrt{M}} \bs S^\top \bs{\bar{S}}^{-1},\quad \bs{\bar S} \coloneqq \diag(\bs{\bar s}) \in \mathbb R^{N\times N},
  \label{eq:Phi-def}
\end{equation}
and $\bs {\bar s} \coloneqq (\bar s_1, \dots, \bar s_N)^\top \in \bb R^N_+$ is the discrete representation (over the pixel grid $\cl X$) of the mean speckle field $\bar S$ defined in~\cref{eq:Sbar-def}. The rationale of the renormalized model~\cref{eq:discrete_model_CS} follows the conclusion of~\cref{sec:speckles}. Since by construction each row $\bs{\bar S}^{-1} \bs s_k$ of the random matrix $\sqrt M \bs \Phi$ in~\cref{eq:Phi-def} corresponds to sampling $S(\bs x; \bs \alpha_k)/\bar S(\bs x)$ on the pixel grid $\cl X$ (provided that $\bar S(\bs x)$ evolves slowly compared to $S$, which basically means that $1/\sigma_c \gg 1/D$),~\cref{prop:subexp_Sx_tildeR} shows that each $\Phi_{jk}$ are sub-exponential for a random speckle illumination generated with random complex amplitudes.

However, the entries of $\bs \Phi$ are biased: their mean is $1/\sqrt{M}$ since $\bar S = \bb E S$. Such a bias is detrimental to the application of CS theory; for instance, it prevents $\bs \Phi$ to satisfy the RIP property~\cite{Rauhut2012}. Following~\cite{raginsky2010compressed,Sudhakar2018}, we can consider the following \emph{debiased observation model} that amounts to modify both the observations and the sensing model:
\begin{equation}
	\bs z \coloneqq \bs y -  (\bs{\bar s}^\top \bs f) \bs 1_M = \sqrt{M}\bs{\tilde  \Phi} \bs{\bar S}\bs f + \bs n,
	\label{eq:noiseless_zeromean_model}
\end{equation}
where $\bs{\tilde \Phi}$ is the debiased sensing matrix
\begin{equation}
\ts \bs{\tilde  \Phi} \coloneqq \bs \Phi - \frac{\bs 1_M \bs 1_N^\top}{\sqrt{M}} = \tfrac{1}{\sqrt{M}} \bs S^\top \bs{\bar{S}}^{-1} - \frac{\bs 1_M \bs 1_N^\top}{\sqrt{M}} \bs{\bar{S}} \bs{\bar{S}}^{-1} = \tfrac{1}{\sqrt{M}} (\bs S^\top  - \bs 1_M \bs{\bar s}^\top) \bs{\bar{S}}^{-1}.
\label{eq:phi_phitilde}
\end{equation}
In this new model, we thus collect $M$ observations $\bs z$ of the object $\bs f$ thanks to the sensing matrix $\bs{\tilde  \Phi} = (\bs{\tilde r}_1, \dots, \bs{\tilde r}_M)^\top/\sqrt{M} \in \bb R^{M \times N}$. By construction, its random entries are zero-mean. Moreover, comparing~\cref{eq:phi_phitilde} to~\cref{eq:final_Sx}, we see that each row $\bs{\tilde r}_k = \bs{\bar S}^{-1}(\bs s_k - \bs{\bar s})$ of $\sqrt M \bs{\tilde  \Phi} = (\bs{\tilde r}_1, \dots, \bs{\tilde r}_M)^\top$ corresponds---for a slowly varying mean field $\bar S$---to the discrete representation of $\tilde R(\bs x; \bs \alpha_k)$ as defined in~\cref{eq:final_Sx}. Therefore, the sub-exponential norm of the entries of $\bs{\tilde \Phi}$ are bounded by $1/\sqrt{M}$.

However, the debiased model~\cref{eq:noiseless_zeromean_model} is impractical; the quantity $(\bs{\bar s}^\top \bs f) \bs 1_M$ is unknown. Up to a slight increase of the noise level, we can solve this problem by assuming that the noise $\bs n$ corrupting the observations is zero-mean. In this case, since $\bb E \frac{1}{M} \bs 1_M^\top \bs y = \bb E \frac{1}{M} \sum_{i=1}^M y_i = \frac{1}{M} \sum_{i=1}^M \bb E(\bs s_i^\top \bs f + n_i) \approx \bs{\bar s}^\top \bs f$, where the expectation refers to the randomness of both $\bs n$ and all the $\{\bs s_i\}_{i=1}^M$, we can reach this realistic sensing model
\begin{equation}
  \ts \bs{\tilde  y} \coloneqq \bs y - \frac{1}{M} (\bs 1_M^\top \bs y) \bs 1_M = \sqrt{M}\bs{\tilde  \Phi} \bs{\bar S}\bs f + \bs{\tilde n}, 
  \label{eq:noiseless_zeromean_model_practice}
\end{equation}
with $\bs{\tilde n} \coloneqq \bs n + \bs n'$ and $\bs n'  \coloneqq (\bs{\bar s}^\top \bs f) \bs 1_M - \frac{1}{M} (\bs 1_M^\top \bs y) \bs 1_M$. Since $\bb E \frac{1}{M} (\bs 1_M^\top \bs y) = \bs{\bar s}^\top \bs f$ and each $y_i$ are independent \rv{s}, we can compute that the variance of the entries of $\bs n'$ decay like $\cl O(1/M)$ when $M$ increases. Therefore, the variance $\tilde \sigma^2$ of each components of $\bs{\tilde n}$ behaves like $\tilde \sigma^2 = \sigma^2 + \cl O(1/M)$, which is close to $\sigma^2$, the variance of each $n_k$, if $M$ is large.   

In the model \cref{eq:noiseless_zeromean_model_practice}, we have explained above that the entries of $\bs{\tilde  \Phi}$ are zero-mean and sub-exponential. For matrices $\bs A$ with zero-mean, random \iid entries with bounded sub-exponential norm, Adamczak \emph{et al.} showed that, provided  $M = \mathcal O(s \ln^2{(N/s)})$, such matrices satisfy the RIP property with high probability when normalized by $\sqrt{M}$ \cite{Adamczak2011}. This means that there exists a constant $0<\delta<1$ such that, for all $s$-sparse vectors $\bs u$ in $\bb R^N$---with at most $s$ non-zero components---we have 
\begin{equation*}
  \ts (1-\delta) \|\bs u\|^2 \leq  \|\frac{1}{\sqrt{M}} \bs A \bs u\|^2 \leq (1+\delta) \|\bs u\|^2_2. 
\end{equation*}
Respecting the RIP implies that robust reconstruction of any $s$-sparse vector $\bs x$ from its (possibly noisy) observations $\bs A \bs x$ can be achieved via, \eg $\ell_1$-minimization or greedy methods~\cite{Rauhut2012}. Foucart \emph{et al.} showed that such matrices also satisfy a modified version of the RIP based on $\ell_1$-norm \cite{Foucart2017} and propose a recovery algorithm able to recover $s$-sparse vectors with $M$ in $\mathcal O(s \ln{(N/s)})$.

Unfortunately, the analysis of the autocorrelation of $\tilde R$ made in~\cref{sec:speckl-illum} shows that all entries of $\bs{\tilde \Phi}$ are (locally) correlated. This is no surprise; while each row $\bs{\tilde r}_k$ of $\sqrt M \bs{\tilde  \Phi}$ belongs to $\bb R^N$, only $J$ parameters---the $J$ random phases of the complex amplitudes $\bs \alpha_k$---were used to generate the pattern $\bs{\tilde r}_k$. Thus, provided that $N \geq J$ (\eg if $J=\cl O(100)$ and $N=\cl O(128^2)$ as in \cref{sec:UCS-Simulations}), spatial correlations are inherent in each $\bs{\tilde r}_k$. Therefore, we cannot readily use the results of \cite{Adamczak2011,Foucart2017} to characterize our sensing scheme. Note that the existence of the above-mentioned correlations is not an impossibility per se; after all, random partial Fourier sensing matrices, made of $M$ randomly sampled rows of a Fourier matrix, do present such correlations but respect the RIP under certain conditions~\cite{Candes2008}. Studying if our sensing matrix $\bs{\tilde \Phi}$ satisfies the RIP is thus an open question.  

There also exists a limit to the resolution achievable by our sensing model. Under the far-field approximation, \cref{eq:resid-field-F-app} shows that 
$$
\ts \tilde R(\bs x; \bs \alpha) \approx \tilde R_{\rm ff}(\bs x; \bs \alpha) \coloneqq \tfrac{1}{J} \sum^J_{j,k=1} \alpha_j\alpha_k^* e^{\frac{2\pi\im}{\lambda z}  (\bs q_j - \bs q_k)^\top \bs x} - 1 = \cl F\big[ \Theta(\bs \cdot; \bs \alpha) \big](- \frac{2\pi}{\lambda z} \bs x),
$$
with ${\Theta(\bs u; \bs \alpha) \coloneqq \tfrac{1}{J} \sum^J_{j,k=1} \alpha_j\alpha_k^* \delta(\bs u - (\bs q_j - \bs q_k) ) - \delta(\bs u)}$. With $D = \max_{j,k} \|\bs q_j - \bs q_k\|$ being the smallest length such that $\Theta(\bs u; \bs \alpha) = \Theta(\bs u; \bs \alpha) \cdot \disk(\frac{\bs u}{D})$, we find
$$
\ts \tilde R_{\rm ff}(\bs x; \bs \alpha) =  (\tilde R_{\rm ff}(\bs \cdot; \bs \alpha) \ast H) (\bs x), \quad\text{with}\  H(\bs x) = H(-\bs x) = (\frac{2\pi D}{\lambda z})^2 \cl F[\disk](\frac{2\pi D}{\lambda z} \bs x).
$$
Therefore, for any row $\bs{\tilde r}_k$ of $\sqrt M \bs{\tilde \Phi}$, which discretizes $\tilde R(\bs x; \bs \alpha_k) \approx \tilde R_{\rm ff}(\bs x; \bs \alpha_k)$, there exists a symetric filter $\bs h$ (the discretization of $H$) whose size scales like $\cl O(\lambda z /D)$ and such that $\bs{\tilde r}_k = \bs{\tilde r}_k * \bs h$. In other words, as also stressed in~\cref{sec:speckl-illum}, the patterns $\bs{\tilde r}_k$ cannot have faster variations than those allowed by the spectrum of $\bs h$; the components of this pattern are correlated and display ``speckle grains'' whose size are related to the size of $\bs h$.

The existence of $\bs h$ alters the model \cref{eq:noiseless_zeromean_model_practice}; since $\scp{\bs a * \bs h}{\bs b} = \scp{\bs a}{\bs b * \bs h}$ for any vector $\bs a$ and $\bs b$, we find 
\begin{equation}
  \ts \bs{\tilde  y} = \sqrt{M}\bs{\tilde  \Phi} (\bs h \ast \bs{\bar S}\bs f) + \bs{\tilde n}. 
  \label{eq:noiseless_zeromean_model_practice_blurred}
\end{equation}
In conclusion, even if $\bs{\tilde  \Phi}$ in this last model was a random matrix with \iid sub-exponential entries, \cref{eq:noiseless_zeromean_model_practice_blurred} shows that, at best, we can estimate $\bs h \ast \bs{\bar S}\bs f$---a version of $\bs{\bar S}\bs f$ that is blurred by $\bs h$. Since the size of this filter is controlled by $\lambda z /D$, the impact of $\bs h$ can be mitigated by either decreasing distance $z$ between the distal end of the fiber and the object (while still keeping the far-field regime valid) or increasing $D$ when designing the MCF. Model \cref{eq:noiseless_zeromean_model_practice_blurred} also shows that the pixel pitch $\varpi$ used to discretize $f$ can be adjusted to the speckle grain $\lambda z /D$ since a higher resolution cannot be achieved. We will follow this procedure in \cref{sec:method} for actual LE fluorescence imaging.

In the following sections, we use the physical model~\cref{eq:discrete_model_03} for data simulations. For the reconstruction process, we consider debiased observations $\bs{\tilde y}$ in \cref{eq:noiseless_zeromean_model_practice}, the associated sensing matrix $\bs{\tilde \Phi}$, and we adapt the CV optimization method \cref{eq:dist_regularized_tilde_TGV_CV} of \cref{sec:choice_parameters} into
\begin{equation}
	\bs{\hat f}_\rho \in \argmin_{\bs u}~\|\bs R_\text{est}(\sqrt{M}\bs{\tilde \Phi}\bs{\bar S}\bs u - \bs{\tilde y})\|^2 + \rho g_1(\bs u) + \imath_{\R{N}_+} (\bs u),
	\label{eq:dist_regularized_tilde_TGV_CV_debiased}
\end{equation}
with quality criterion $\tilde q(\bs{\hat f}_\rho) \coloneqq \|\bs R_\text{val}(\sqrt{M}\bs \Phi \bs{\bar S} \bs{\hat f}_\rho - \bs{\tilde y})\|^2$. As will be clear below, the inclusion of $\bs{\bar S}$ in this minimization---induced by inserting \cref{eq:noiseless_zeromean_model_practice} in \cref{eq:dist_regularized_tilde_TGV_CV}---allows us to expand the area where $\bs f$ is estimated outside of the support of $\bs{\bar S}$.  

As explained in~\cref{sec:UCS-Simulations}, this optimization scheme leads to a quality of the final estimate similar to the one obtained with the biased sensing model but achieved in a much shorter reconstruction time.

\subsubsection{Simulations}
\label{sec:UCS-Simulations}
We simulate observations $\bs y$ of $128\times 128$ USAF transmission target $\bs f$ (see~\cref{fig:USAF_GT}) according to~\cref{eq:discrete_model_03} for $M/N\in\{0.1,0.2,\dots,1\}$. We reconstruct estimates of $\bs f$ by first solving the original inverse problem~\cref{eq:dist_regularized_tilde_TGV_CV} with $(\bs \Phi,\bs y)$ and then the debiased model~\cref{eq:dist_regularized_tilde_TGV_CV_debiased} with $(\tilde{\bs\Phi},\bs{\tilde y})$. For comparison, we also simulate RS observations and measurements acquired with an ideal zero-mean Gaussian matrix $\bs\Phi_{\rm G}$. Since the synthetic USAF transmission target is piecewise constant, we select the TV-norm for $g_1$ in~\cref{eq:dist_regularized_tilde_TGV_CV,eq:dist_regularized_tilde_TGV_CV_debiased}.

\captionsetup[subfigure]{labelformat = parens}
\begin{figure}
	\centering
	\subfloat[\label{fig:USAF_GT}]{\includegraphics{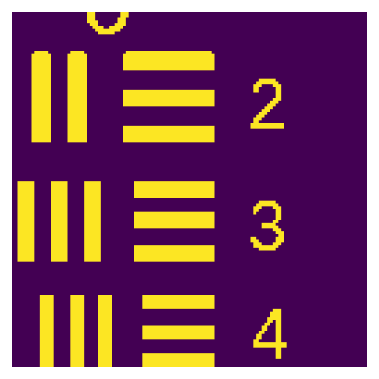}}
	\subfloat[\label{fig:USAF_GT_vignetted}]{
          \includegraphics{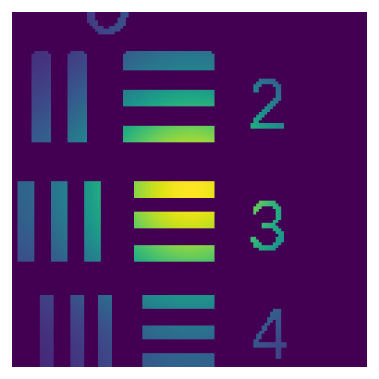}	}
	\subfloat[\label{fig:USAF_RS}]{
          \includegraphics{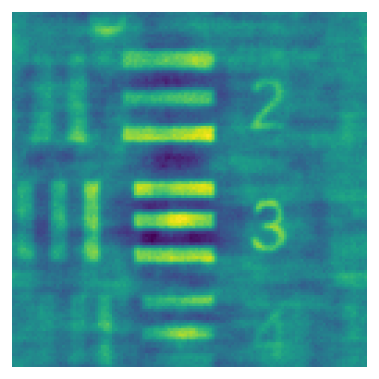}	}
	\caption{(a) $128\times 128$ ground truth $\bs f$ of standard USAF transmission target. (b) Vignetted ground truth $\bar{\bs S} \bs f$. (c) Acquisition with RS strategy as explained in \cref{sec:estimation_f}. Observations were generated according to \cref{eq:discrete_model_03} with focused illumination pattern and AWGN (SNR of $\bs\Phi \bs f$ equal to $40$ dB). The maximum intensity of the RS image is around three times lower than the one of the ground truth.}
	\label{fig:USAF_GT_RS}
\end{figure}

We measure the quality of an estimate $\bs{\hat u}$ with both the signal-to-noise ratio (SNR) and the weighted SNR metrics: 
\begin{equation*}
	\ts \SNR(\bs{\hat u},\bs u) \coloneqq 20\log_{10}{\big(\frac{\|\bs{\hat u}\|}{\|\bs{\hat u}-\bs u\|}\big)},\quad \WSNR(\bs{\hat u},\bs u) \coloneqq \SNR(\bs{\bar S}\bs{\hat u},\bs{\bar S}\bs u).
\end{equation*}
The WSNR attenuates the reconstruction artefacts at the limit of the FOV where the vignetting of $\bs{\bar S}$ is the strongest.

For each value of $M$ above, $M$ observations of $\bs f$ are generated according to model~\cref{eq:discrete_model_03}, with a noise variance $\sigma^2$ set such that the SNR of $\bs\Phi\bs f$ is equal to 40\,dB. For all the experiments, the illumination patterns were generated for a Fermat's spiral core arrangement ($J=120$ cores and diameter $D=113\,\mu$m) and the following parameters: $\lambda=1\,\mu$m, $z=500\,\mu$m, $3 \sigma_c = 3.2\,\mu$m and pixel pitch $\varpi = 2\,\mu$m.

We perform the reconstruction with the CP algorithm (see~\cref{sec:algo}) and a maximum number of internal iterations equal to $5\,000$, never reached in practice.  The stopping criterion and regularization parameter are set as described in~\cref{sec:choice_parameters} with $M_\text{val} = 256$, and initial values $\rho^{(0)} = 10$ and $\bs{\hat f}_{\rho^{(0)}} = \bs 0_N$. The initial guess of the algorithm for each $k> 0$ is given by previous estimate $\bs{\hat f}_{\rho^{(k-1)}}$ and $\rho^{(k)} = 0.5\rho^{(k-1)}$. The maximum number of iterations on the estimation of $\rho$ is $20$.

\captionsetup[subfigure]{labelformat=empty}

\begin{figure}
	\centering
	\subfloat{
          \includegraphics{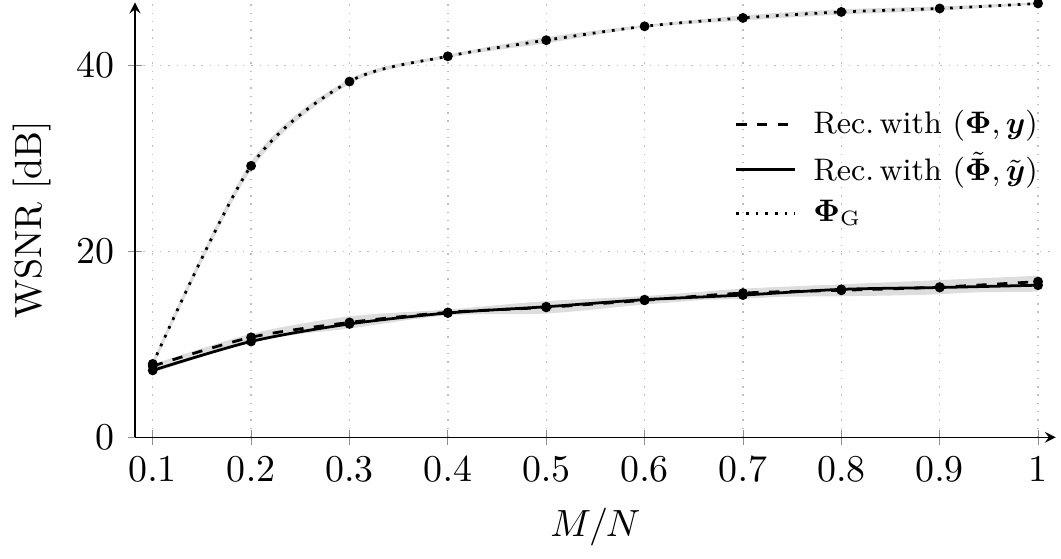}	}
	\\
	\subfloat{\raisebox{0.1\linewidth}{\rotatebox[origin=c]{90}{\small$\bs{\hat f}$}}
	}
	\subfloat[Rec.\,with $({\bs \Phi},{\bs y})$]{
\includegraphics{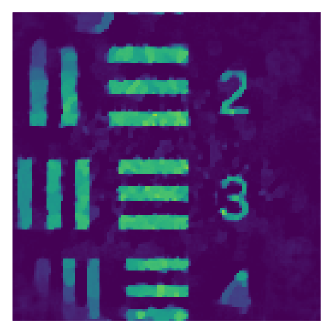}	}
	\subfloat[Rec.\,with $(\tilde{\bs \Phi},\tilde{\bs y})$]{
\includegraphics{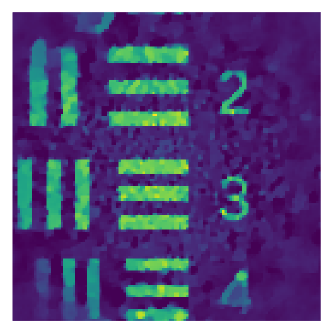}	}
	\subfloat[$\bs \Phi_{\rm G}$]{
\includegraphics{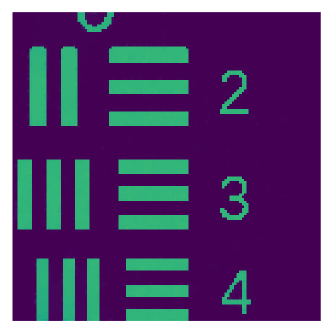}	}
	\caption{\label{fig:noise40_UCS_comparison}\textbf{SI strategies.} Above: mean WSNR (over 10 trials) of the restored USAF target \emph{versus} $M/N$ for three SI strategies: (\emph{i}) acquisition and reconstruction with $({\bs \Phi},{\bs y})$ (dashed line), (\emph{ii}) acquisition with ${\bs \Phi}$ and reconstruction with $(\tilde{\bs \Phi},\tilde{\bs y})$ (solid line), and (\emph{iii}) acquisition and reconstruction with ideal (unrealistic) sensing matrix $\bs \Phi_{\rm G}$ (dotted line). The gray areas represent the standard deviations. Synthetic observations were generated according to \cref{eq:discrete_model_03} (SNR of $\bs \Phi \bs f$ equal to $40$\,dB). Below: estimates of $\bs f$ obtained for $M/N=0.3$. All estimates share the same intensity scale.}
      \end{figure}

\subsubsection{Comparison results}
Simulation results are visible on~\cref{fig:noise40_UCS_comparison}. Estimates of the fluorophore density map obtained with SI methods have a better quality compared to the estimate obtained with RS (WSNR equal to 1.68\,dB with optimal normalization, see~\cref{fig:USAF_RS}). However, unlike RS, SI requires to store $M$ speckles patterns needed for the reconstruction.

Acquisition and reconstruction with a Gaussian sensing matrix outperforms the other two frameworks involving speckle illumination. While unrealistic, such a Gaussian framework is useful as it sets an upper bound on the achievable reconstruction quality of LE imaging. For instance, it shows that the CS regime, where enough measurements are collected to ensure a reconstruction quality close to the measurement SNR of 40\,dB, starts around $M/N = 0.3$ after a sharp transition in the curve. After that value, an increase in the number of measurements does not lead to a significant gain in WSNR. 
The same behavior is observed for reconstructions based on speckle acquisition even if the change in the curve rate is less outstanding. 
WSNR of images reconstructed with $(\tilde{\bs\Phi},\bs{\tilde y})$ is similar to the one obtained for a reconstruction performed with $({\bs\Phi},{\bs y})$, despite the extra noise introduced in the debiased model. However, we observed significantly shorter reconstruction times for the debiased model when $M/N\le 0.7$. We also noticed that, while the reconstruction with $({\bs\Phi},{\bs y})$ fails when we do not constraint the estimate to be non negative, the debiased model performs almost as well without enforcing the positivity of the map. In the rest fo the paper, all reconstructions are performed with the debiased model.

\subsection{Partial Speckle Scanning}
\label{sec:SCS}

One disadvantage of SI is that, despite the good quality of its reconstruction, it requires a long acquisition time compared to RS imaging that simply uses scanning mirrors. This is due to the strong contrast that exists between the time it takes to change the SLM configuration in SI---about $1/60\,$s in our setup \cite[see DVI frame rate]{HamamatsuSLM}---and the elapsed time between two consecutive tilts of the galvanometric mirrors, which is about 1,000 times faster~\cite{CambridgeTechnologyInc.2019}.
  
We here propose a hybrid acquisition framework, called \emph{partial speckle scanning} (PSS), combining the advantages of both techniques while keeping a high reconstruction quality. This method strongly relies on two properties of the LE considered in this paper: (\emph{i}) the ability to easily generate speckles by randomly programming the SLM and (\emph{ii}) the MCF memory effect.

The PSS strategy acquires $M$ observations but unlike the SI strategy, a single SLM configuration allows us to collect $M_{P}\le M$ observations by translating the speckle. Following~\cref{sec:illum-transl}, this is achieved by applying different tips to the input wave front to the MCF with the scan mirrors. Mathematically, given a set of $P=M/M_P$ complex amplitudes $\{\bs \alpha_j\}_{j=1}^{P}$ randomly generated as in the SI model, and $M_{P}$ mirror tilts $\{\bs \theta_k\}_{k=1}^{M_P}$, the PSS sensing matrix $\bs{\tilde \Phi}$ used in the model (\ref{eq:noiseless_zeromean_model_practice}) corresponds to
  \begin{equation}
    \label{eq:S-CS-model}
    \ts  \bs{\tilde \Phi} = [\bs{\tilde \Phi}^\top_1,\, \cdots, \bs{\tilde \Phi}^\top_{P}]^\top,\quad\text{with}\ \bs{\tilde \Phi}^\top_j = [\bs{\tilde r}^{(0)}_j, \bs{\tilde r}^{(1)}_j,\, \cdots, \bs{\tilde r}^{(M_{P}-1)}_j]^\top.    
  \end{equation}
  and where, according to \cref{eq:tildeR_translation}, $\bs{\tilde r}^{(k)}_j$ is the discrete representation of $\tilde R\big(\bs x, \diag(\bs \gamma(\bs \theta_k)) \bs \alpha_i \big) \approx \tilde R(\bs x + \bs \theta_k, \bs \alpha_i)$---the shifted residual field---for each $j \in [P]$, $k \in [M_P]$, and $\bs \gamma$ defined in \cref{eq:gamma-mod-def}.

  In this work, we consider a single line scanning mode, \ie shifts are applied to the speckle patterns in only one (abitrary) unit direction $\bs u \in \bb R^2$ and $\bs \theta_k = k \delta\, \bs u$ for some translation step $\delta >0$ of the speckle patterns in the plane $\cl Z$. Line scanning is fast and accurate because it only needs the rotation of one galvanometric mirror.

When designing the PSS framework, the adjustment of the shift $\delta$ between two illumination patterns (which we perform in the next section) faces two competing effects. First, if $\delta$ is smaller than a speckle grain, \ie if too small compared to $\lambda z/D$, the shifting model \cref{eq:S-CS-model} introduces too much correlations between neighbouring rows of $\bs{\tilde \Phi}$ and there is not much variety in the $M_P$ observations acquired with a fixed SLM configuration. Moreover, in addition to the column dependency mentioned in~\cref{sec:UCS}, our sensing matrix further deviates from an ideal decorrelated sub-exponential random matrix. Second, to approximate $\tilde{\bs \Phi}$ by a block-circulant matrix (\eg to boost the computation of matrix-vector multiplication), $\delta$ must be small enough such that approximation~\cref{eq:Sx_tilde_translation} still holds. This imposes us to respect the paraxial approximation for all translation steps, \ie we must have $M_P \delta \ll z$. Note that, while there exist sensing constructions based on subsampled random circulant matrices defined from sub-Gaussian random filters \cite{Duarte2011,Rauhut2012}, there are no known constructions for sub-exponential random filters. Despite this absence, our simulations below confirm that the PSS sensing compares favorably to the SI and the Gaussian sensing schemes.

An advantage of the PSS scheme is that its scanning time $t_{\rm acq}$ is reduced compared to the SI acquisition time; a desirable advantage when observing fluorescent biological samples subject to photo-bleaching (see~\cref{sec:fluo-model}). We can compute this time from
\begin{equation}
t_{\rm acq} = t_{\rm galva} (\xi P  + M - P ),\quad\text{with}\ M = P M_P, 
\label{eq:tacq}
\end{equation}
where $1/t_{\rm galva}$ is the scan mirror rate and $\xi \coloneqq t_{\rm SLM}/t_\text{galva}$ with $1/t_\text{SLM}$ the SLM frame rate.
If $t_{\rm SLM}$ is large compared to $t_{\rm galva}$, \ie if $\xi$ is large, for a fixed $M$, $t_{\rm acq}$ can be kept small if we keep the number of distinct speckle patterns $P$ small. In our setup, $\xi \approx 10^3$.

\subsubsection{Simulations}

As a first simulation, we aim to choose a convenient shift $\delta$ between two consecutive speckles. We simulate $M$ observations $\bs y$ of $128 \times 128$ USAF transmission target $\bs f$ for $M/N=0.3$. $M_P=2$ observations are acquired before changing the configuration of the SLM, \ie $P/N = 0.15$. We sequentially set the shift $\delta$ to the values $\{0,1,2,3,4,6,8,10\}$ [$\mu$m]. For each, we reconstruct the estimate $\bs{\hat f}$ of $\bs f$ by solving the debiased model~\cref{eq:dist_regularized_tilde_TGV_CV_debiased} with $(\bs{\tilde \Phi},\bs{\tilde y})$. Exactly like for the SI strategy, $g_1$ in~\cref{eq:dist_regularized_tilde_TGV_CV_debiased} is the TV-norm.

We also conduct a second simulation where we compare the PSS strategies for different values of $M_P$. We simulate and reconstruct observations as described in the previous paragraph but with fixed value of $\delta$ (chosen according to the result of the first experiment), for $M/N\in\{0.1,0.2,\dots,1\}$ and $M_P \in \{1,2,4,8,16,32,64\}$. The value $M_P=1$ corresponds to the SI strategy. We generate the observations according to~\cref{eq:discrete_model_03} with $\sigma^2$ set to reach a measurement SNR of 40\,dB. The speckles are generated for a Fermat's spiral core arrangement ($J=120$ cores and diameter $D=113\,\mu$m) with the following parameters: $\lambda=1\,\mu$m, $z=500\,\mu$m, $3 \sigma_c = 3.2\,\mu$m and a pixel pitch $\varpi = 2\,\mu$m. The time $t_\text{SLM}$ is set to 100 ms, slightly higher than the actual SLM time \cite{HamamatsuSLM}, and $t_\text{galva}$ is set to 100$\,\mu$s \cite{CambridgeTechnologyInc.2019}.
The reconstruction is performed as already described in the SI simulations (see~\cref{sec:UCS-Simulations}).

\subsubsection{Results}

\begin{figure}
	\centering
	\subfloat{
          \includegraphics{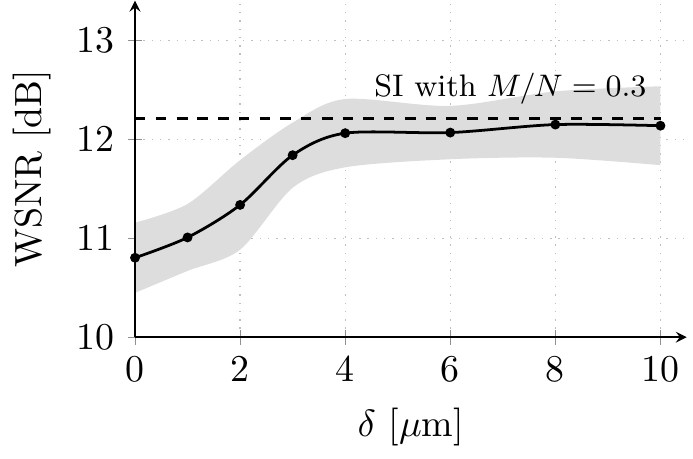}	}\hfill
	\subfloat{
\includegraphics{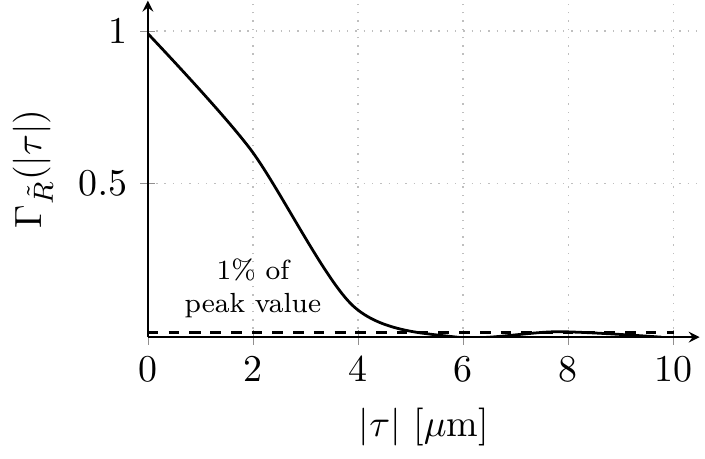}	}
	\caption{\label{fig:SCS_dmin}\textbf{PSS strategy: choice of $\bs \delta$.} (Left) Mean WSNR (over 25 trials) of the restored USAF target \emph{versus} the shift between two consecutive replicas of the same speckle. 
	The gray area represents the standard deviation. Synthetic observations were generated according to \cref{eq:discrete_model_03} (SNR of $\bs\Phi \bs f$ equal to 40\,dB). Each illumination pattern was used twice and $M/N=0.3$. (Right) Theoretical autocorrelation of the residual field $\tilde R$ defined in \cref{eq:gamma_R_final} as a function of $|\tau|$ (extracted from \cref{fig:autocorrelations}).
	}
\end{figure}

Results in~\cref{fig:SCS_dmin} suggest to consider a shift $\delta\ge4\,\mu\text{m}$. In this case, the spatial correlation between shifted patterns (related to the autocorrelation of field $\tilde R$ and consequently to the speckle grain as detailed in \cref{sec:speckl-illum}) is close to zero and the quality of the estimate is similar to the one obtained with the SI strategy. 

\begin{figure}
	\centering
	\subfloat{
\includegraphics{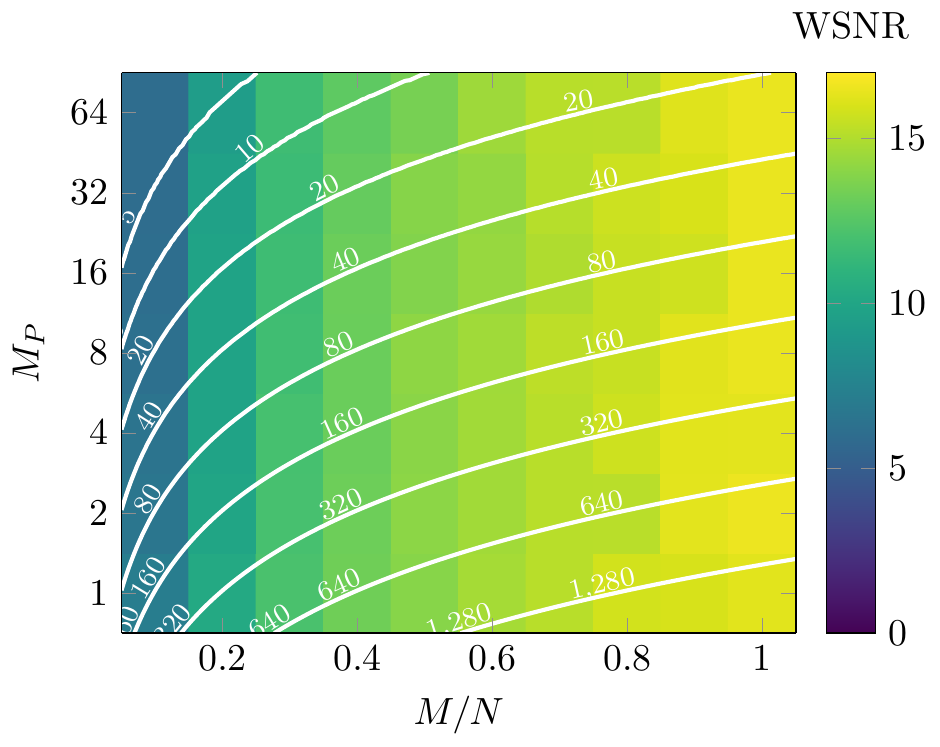}	}
	\\
	\subfloat{\raisebox{0.1\linewidth}{\rotatebox[origin=c]{90}{\small$t_\text{acq}\approx 20$\,s}}
	}
	\subfloat[$M_P = 8$]{
\includegraphics{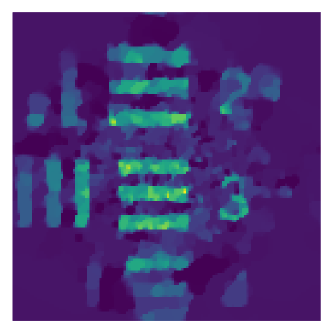}	}
	\subfloat[$M_P = 16$]{
\includegraphics{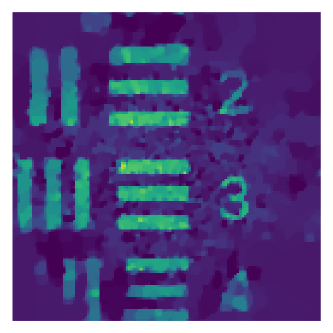}	}
	\subfloat[$M_P = 32$]{
\includegraphics{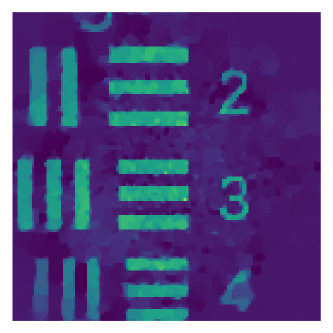}	}
	\subfloat[$M_P = 64$]{
\includegraphics{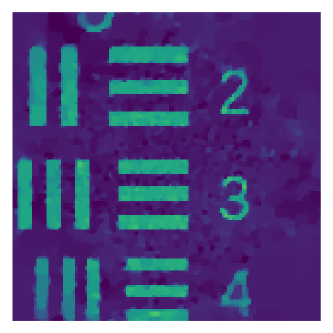}	}
	\caption{\label{fig:WSNR_tacq_level_curves_d2}\textbf{PSS strategy.} Above: mean WSNR (over 10 trials) of the restored USAF target \emph{versus} $M_P$ and $M/N$ ratio. Synthetic observations were generated according to \cref{eq:discrete_model_03} with AWGN (SNR of $\bs\Phi \bs f$ equal to $40$\,dB). Number $M_P$ of replicas of the same (shifted) speckle pattern belongs to $\{1,2,4,8,16,32,64\}$ and shift $\delta$ between replicas is $4\,\mu$m. Level curves of the acquisition time $t_\text{acq}$ (in seconds) are superimposed to the SNR (white solid lines). Below: estimates of $\bs f$ obtained for $M_P \in\{8,16,32,64\}$ and $t_\text{acq}\approx 20$\,s. The PSS strategy with $M_P = 1$ corresponds to the SI strategy with $(\tilde{\bs \Phi},\tilde{\bs y})$ (see \cref{fig:noise40_UCS_comparison}). Estimates share the same intensity scale.}
\end{figure}

\cref{fig:WSNR_tacq_level_curves_d2} shows the results of the second experiment performed with $\delta=4\,\mu$m. Regarding~\cref{fig:SCS_dmin} and the choice of $\delta$, it is not surprising to observe a reconstruction quality similar for all $M_P$ values at fixed $M/N$. However, even if the reconstruction quality is similar for a given $M/N$ ratio, we are more interested in considering the quality obtained for a fixed acquisition time $t_\text{acq}$. To minimize the photo-bleaching of the sample, we would like $t_\text{acq}$ as small as possible while keeping high WSNR. In this case, the best choice is the PSS strategy with $M_P=64$: we achieve high WSNR (around 17\,dB) with an acquisition time around 30\,s.

\begin{figure}
	\centering
	\subfloat{
\includegraphics{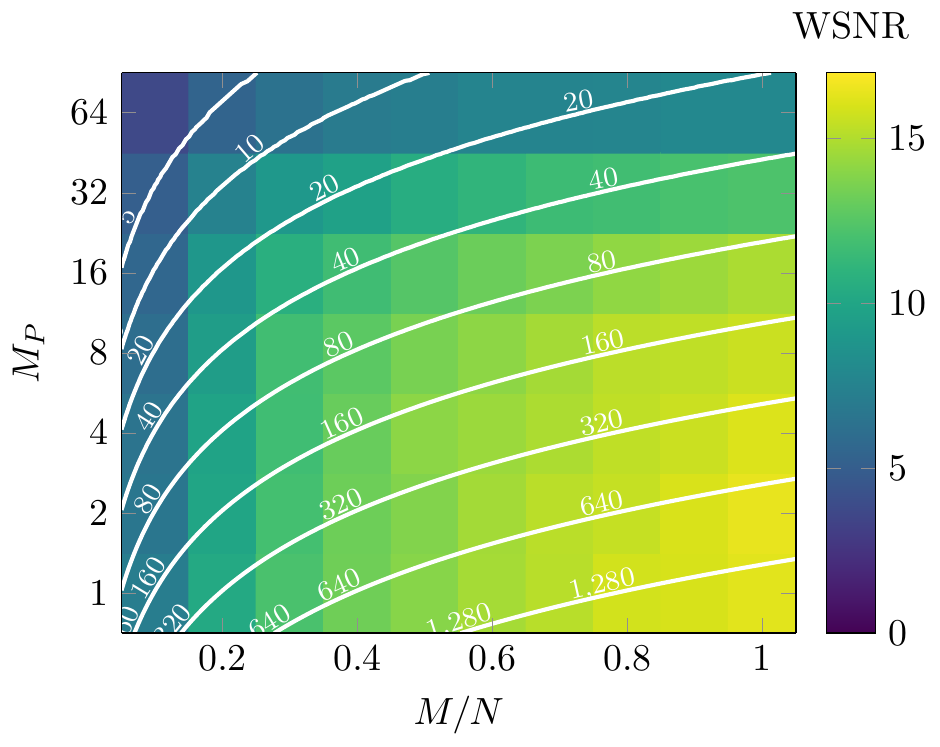}	}
	\\
	\subfloat{\raisebox{0.1\linewidth}{\rotatebox[origin=c]{90}{\small$t_\text{acq}\approx 20$\,s}}
	}
	\subfloat[$M_P = 8$]{
\includegraphics{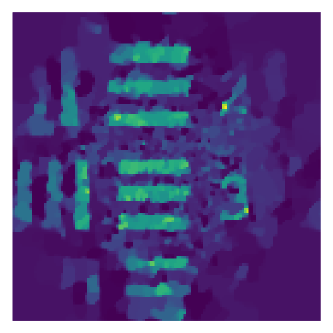}	}
	\subfloat[$M_P = 16$]{
\includegraphics{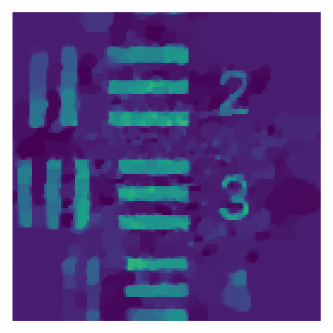}	}
	\subfloat[$M_P = 32$]{
\includegraphics{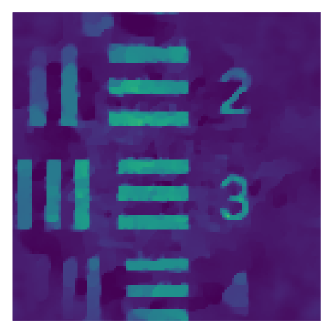}	}
	\subfloat[$M_P = 64$]{
\includegraphics{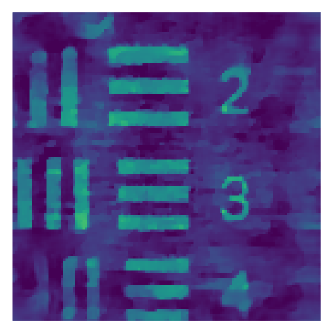}	}
	\caption{\label{fig:WSNR_tacq_level_curves_d2-approx}\textbf{PSS approximated strategy.} Above: mean WSNR (over 10 trials) of the restored USAF target \emph{versus} $M_P$ and $M/N$ ratio. Synthetic observations were generated according to \cref{eq:discrete_model_03} with AWGN (SNR of $\bs\Phi \bs f$ equal to $40$\,dB). Number $M_P$ of replicas of the same (shifted) speckle pattern belongs to $\{1,2,4,8,16,32,64\}$ and shift $\delta$ between replicas is $4\,\mu$m. Sensing matrix $\tilde{\bs \Phi}$ was approximated by a block-circulant operator (instead of a block-Toepliz operator since the FOV is limited) for the reconstruction. Level curves of the acquisition time $t_\text{acq}$ (in seconds) are superimposed to the SNR (white solid lines). Below: estimates of $\bs f$ obtained for $M_P \in\{8,16,32,64\}$ and $t_\text{acq}\approx 20$\,s. The PSS strategy with $M_P = 1$ corresponds to the SI strategy with $(\tilde{\bs \Phi},\tilde{\bs y})$ (see \cref{fig:noise40_UCS_comparison}).}
      \end{figure}

\begin{remark}[Approximation of the sensing matrix]
  \label{rem:approx_phi}
  Up to now, the proposed PSS strategy only decreases the acquisition time. Compared to SI for the same number of measurements, the number of speckles to be recorded to form $\bs{\tilde \Phi}$ is identical. Moreover, the reconstruction time is similar between SI and PSS since the complexity (in $\cl O(MN)$) of the matrix-vector product involving $\bs{\tilde \Phi}$ is not optimized. By considering the translation rule \cref{eq:Sx_tilde_translation}, this can be potentially improved by approximating $\tilde{\bs\Phi}$ with a block-Toepliz matrix $\tilde{\bs\Phi}_{\rm app}$. Then, each matrix-vector product involving this new matrix can benefit from the FFT and the complexity is then reduced to $\cl O(N \log N)$. The previous storage of $M$ speckle patterns is reduced to $P$ patterns. We test this matrix approximation in \cref{fig:WSNR_tacq_level_curves_d2-approx}. This figure is the same as~\cref{fig:WSNR_tacq_level_curves_d2} but for a reconstruction performed with $\tilde{\bs\Phi}_{\rm app}$. As expected, when $M_P$ increases, $\|\bs\theta_{M_P}\| = M_P \delta$ takes bigger values and the paraxial approximation is less respected. The error made by approximation~\cref{eq:Sx_tilde_translation} increases. This leads to a decrease in the reconstruction quality. However, for $M_P\le 16$, the reconstruction quality is still similar to the one of the SI strategy (around $15$\,dB). Therefore, in this approximation of $\bs{\tilde \Phi}$ by $\bs{\tilde \Phi}_{\rm app}$, one must find a trade-off between a fast acquisition time and the quality of the produced estimate. We anyway postpone this analysis for a future study, and we only consider the original matrix $\bs{\tilde \Phi}$ in our following experiments.
\end{remark}

\section{Fluorescence imaging experiments}
\label{sec:experiments}
In this last section, we apply the SI and PSS methods to an actual lensless endoscope in the context of fluorescence imaging. We first describe the experimental setup, material and methods, before explaining how these sensing strategies can improve the quality of the reconstructed images compared to the RS technique.   

\subsection{Experimental setup}
\label{sec:experimental_setup}
A simplified view of the experimental setup is presented in \cref{fig:setup} and described below in functional blocks. 

\begin{figure}
	\centering
\includegraphics{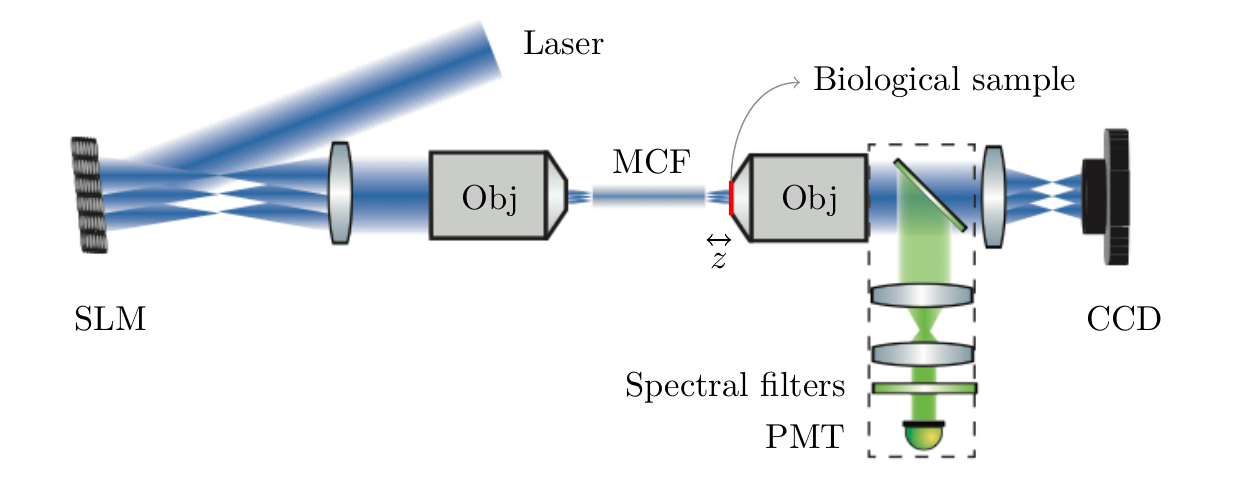}	\caption{\label{fig:setup}Simplified view of the experimental setup used to characterize the performance of the CS-based acquisition strategies. Relay optics (Obj) image the spatial light modulator (SLM) to the proximal endface of the multicore fiber (MCF). An imaging system with two channels is placed at a distance $z$ of the MCF distal end. The main purpose of the first channel is to image the speckle patterns with a charge-coupled device (CCD) sensor. The second channel is dedicated to observations collection through a single pixel detector or photomultiplier tube (PMT).}
      \end{figure}
      
\paragraph{Spatial light modulator}
A continous wave laser operating at 491\,nm (Cobolt lasers, Sweden) is expanded and impinges upon a liquid crystal SLM (X10468-03, Hamamatsu, Japan). A set of relay optics, depicted by a single lens and objectives $Obj$ in \cref{fig:setup}, images the SLM  to the proximal endface of the MCF. In order to maximize the injection of the light beams into the individual cores, a convex lenslet array whose centers are matched to the individual cores is displayed on the SLM. This results in a beamlet array which is efficiently coupled into the MCF and whose relative phases can be tuned independently with the SLM.

\paragraph{Multicore fiber}
 Imaging through the MCF is analogous to phased arrays for beamforming where the relative phases between the cores (antennae) can be calibrated and tuned to generate and shift a focused beam. The cores of these fibers are single mode at the operating wavelength and exhibit an inter-core coupling term less than 20\,dB \cite{Sivankutty2018}. These factors have two important advantages: (\emph{i}) operations such as RS or defocusing can be performed with conventional optical elements at the distal end, and (\emph{ii}) the resulting speckle patterns are highly resilient to external perturbations both thermal and mechanical (except for a global shift). These are significant advantages over single multimode fiber (MMF) where the significant off-diagonal coupling terms preclude stability and fast imaging with conventional optical elements. For a more comprehensive discussion of the imaging properties of the golden spiral MCFs, we refer the reader to \cite{Sivankutty2018}.

\paragraph{Generation of the speckle patterns}
Multiple illumination patterns are generated a few hundred microns away from the distal end as a combination of (\emph{i}) randomizing the relative phase of the injected beamlets into the MCF  resulting in a speckle pattern at the distal end of the fiber, and (\emph{ii}) translations  of the speckle pattern with a global tilt of the beamlets.

\paragraph{Calibration and imaging}
The distal end of the MCF is placed at the focal plane of a second imaging system with two channels: one imaging the distal end onto the camera (CCD) and a second one detecting the signal of interest (PMT). The first channel serves for the recording of the sensing matrix, the visual inspection of the samples and the generation of an image close to the ground truth. The sensing matrix is populated by acquiring a multi-exposure image of each speckle pattern and fusing them to generate a synthetic high-dynamic range. 
The second channel employs a single pixel detector (PMT) (H7240-50, Hamamatsu, Japan) upon which the signal is detected. In the following experiments, the sample is either a standard USAF transmission target or fluorescent beads. In either case, we operate in a high photon count regime as assumed in \cref{sec:acquisition_model}.

In the interest of maximum flexibility, we employ the SLM to function as a series of optical components such as a microlens array to maximize the coupling into the fiber, and as a galvanometric scanner. This also allows us to test and compare between  the proposed imaging scheme and conventional RS techniques. However, for the speckle illumination based compressive imaging proposed in the paper, the SLM can be replaced with conventional optical elements such as microlens arrays, mirror scanners and thin diffusers as in our earlier works~\cite{Andresen2013}. 

\subsection{Material and methods}
\label{sec:method}

Samples used for the experiments are either standard \texttt{USAF} transmission target (see \cref{fig:exp_GT}) or 5\,$\mu$m fluorescent \texttt{Beads}. \texttt{USAF} sample itself is not fluorescent: it is a layer of metal deposited on glass where the features (the bars) are transparent. To make it fluorescent, we apply a layer of highlighter on top of it. We use the setup depicted in \cref{fig:setup} and described in \cref{sec:experimental_setup} with a Fermat's golden spiral MCF containing $J=118$ cores and with standard deviation $\sigma_c = 0.8\,\mu\text{m}$ (or, equivalently, $d = 1.9\,\mu\text{m}$). The distance between the endface of the fiber and the sample plane is $z = 500\,\mu\text{m}$.

\begin{figure}
	\centering
	\subfloat[\texttt{USAF \#1}]{
\includegraphics{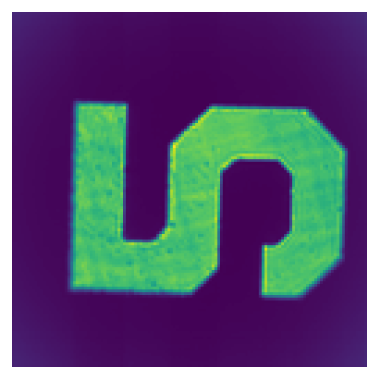}	}
	\subfloat[\texttt{USAF \#2}]{
\includegraphics{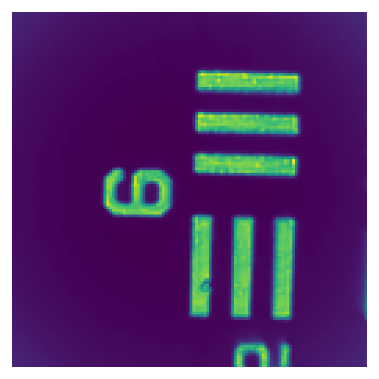}	}
	\subfloat[Mean speckle field]{
\includegraphics{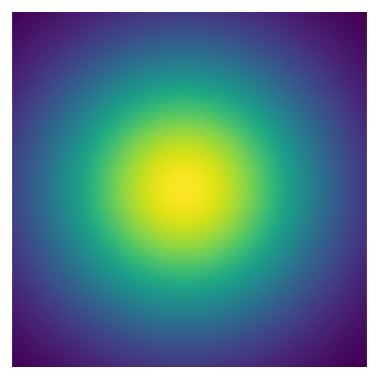}	}
	\caption{\textbf{Estimates of \texttt{USAF} ground truth.} Images of standard USAF transmission target. They are estimated using the first channel of the imaging system (see \cref{fig:setup}). The CCD sensor acquires $M$ images of the product between a speckle and density map $f$. Their average divided by the mean speckle field (Gaussian fit estimated from the $M$ light patterns) leads to the estimated ground truths.}
	\label{fig:exp_GT}
      \end{figure}
      
\subsubsection{Acquisition and reconstruction}
In all experiments, the data acquisition follows those two steps: (\emph{i}) recording the speckle patterns with the CCD sensor to build and store the sensing matrix $\bs \Phi$ (pixel size equal to $2.2\,\mu$m), and (\emph{ii}) illuminating the sample with $M$ speckles and measuring the signal on the single pixel detector. Light patterns are either all different from each other (SI strategy) or shifted versions of each other (PSS strategies) due to the application of global tilt terms on the SLM ($\delta=1.1\,\mu$m). 

The first experiment was designed to compare the SI and PSS strategies: $M=4096$ measurements of \texttt{Beads \#3} sample were acquired for $M_P=1$ (SI) and then $M_P=64$ (PSS).
The second experiment acquired $M=4096$ observations of two different parts of the \texttt{USAF} target (see \cref{fig:exp_GT}) with $M_P = 64$. Finally, $M=4096$ observations of three other \texttt{Beads} samples were acquired with $M_P=64$.

Unlike the number of observations $M$ that is an acquisition parameter, the number of pixels $N$ of the estimate (or equivalently, the pixel pitch $\varpi$) can be chosen after the acquisition process. Ideally, $\varpi$ should match the diffraction limited point spread function of the device, \ie the speckle grain size. In this case, we can take full advantage of the CS-friendly statistical properties of the sensing matrix and we avoid spurious correlations in the estimate. 
For the considered experimental parameters, the average grain size is $r \approx 3.5\,\mu$m (see \cref{sec:speckl-illum}).
We thus set the pixel pitch $\varpi$ to $r$, and select $N=80\times 80$ for \texttt{USAF} samples and $N=160\times160$ for \texttt{Beads} samples (original sizes are $128\times 128$ and $256\times 256$, respectively).

We obtained better results by solving a slightly different problem from \cref{eq:dist_regularized_tilde_TGV_CV_debiased}: given the zero-mean sensing matrix $\tilde{\bs \Phi}$ and the debiased observations $\tilde{\bs y}$, we solve
\begin{equation}
	\widehat{\bs{\bar S}\bs f}_\rho \in \underset{\bs u}{\text{arg~min}}~\|\bs R_\text{est}(\sqrt{M}\tilde{\bs \Phi}\bs u - \tilde{\bs y})\|^2 + \rho g_1(\bs u) + \imath_{\R{N}_+} (\bs u),
	\label{eq:dist_regularized_tilde_TGV_CV_debiased_exp}
      \end{equation}
\ie we estimate the \emph{vignetted} fluorophore density map $\bs{\bar S}\bs f$ instead of $\bs f$. This adaptation is possible by setting $g_1$ in \cref{eq:dist_regularized_tilde_TGV_CV_debiased_exp} to the $\TGV$-norm. This norm, which promotes piecewise linear images, is well adapted to our piecewise constant images multiplied by the smooth vignetting~$\bs{\bar s}$. As in the synthetic case, we use the CP algorithm with a maximum number of iterations equal to $5000$ (that was also never reached). The stopping criterion of CP internal iterations is set as described in \cref{sec:choice_parameters} with $M_\text{val} = 256$, while the parameter $\rho$ is chosen by visual inspection. Initial values for $\rho$ are $\rho^{(0)} = 2500$ (\texttt{USAF}) or $\rho^{(0)}  = 10^4$ (\texttt{Beads}) and $\rho^{(k)} = 0.8\rho^{(k-1)}$ for $k>0$ ($k_{\max} = 20$). Initial guess of the algorithm is $\hat{\bs f}_{\rho^{(0)}} = \bs 0_N$ and then $\hat{\bs f}_{\rho^{(k-1)}}$ for $k> 0$.

\subsection{Results}
\label{sec:exp_results}
\begin{figure}
	\captionsetup[subfigure]{justification=centering}
	\centering
	\subfloat[SI Beads 3][\textbf{SI}\\$M = 4096$\\$t_\text{acq} = 409.6$\,s]{
\includegraphics{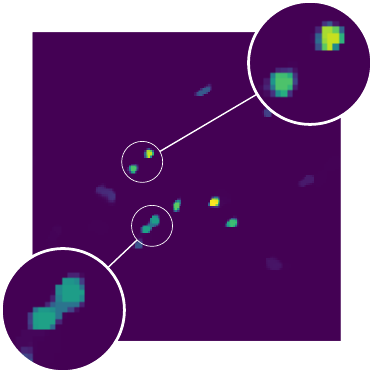}	}
	\subfloat[PSS Beads 3][\textbf{PSS}\\$M = 4096$\\$t_\text{acq} = 6.8$\,s]{
	\includegraphics{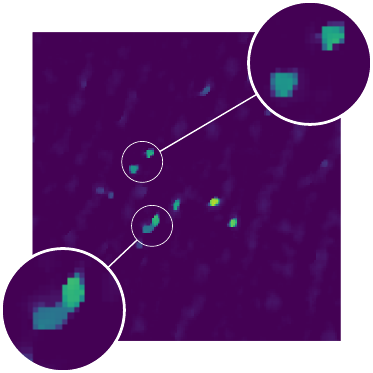}}
	\subfloat[SI Beads 3][\textbf{SI}\\$M = 68$\\$t_\text{acq} = 6.8$\,s]{
          \includegraphics{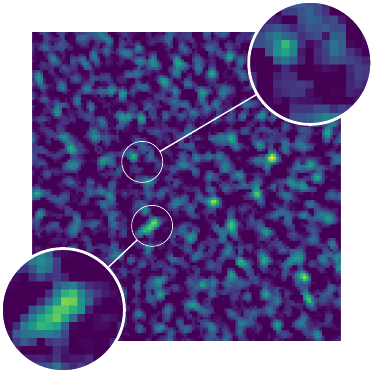}	}
	\caption{\label{fig:exp_results_beads_U-CS}\textbf{Vignetted estimates of \texttt{Beads \#3} sample.} Left and middle images are obtained with two different acquisition strategies (SI and PSS) but the same number of observations $M$ ($\rho=860$ and $\rho = 690$). Middle and right images are again obtained with two different strategies (PSS and SI) but the same acquisition time $t_{\text{acq}} = 6.8\,\text{s}$. Left and middle estimates share the same intensity scale. The intensity of the right estimate is 50 times lower compared to the other two images.}
      \end{figure}

\begin{figure}
	\centering
	\subfloat{\raisebox{0.1\linewidth}{\small\texttt{\#1}}
	}
	\subfloat[$\rho=420$]{
\includegraphics{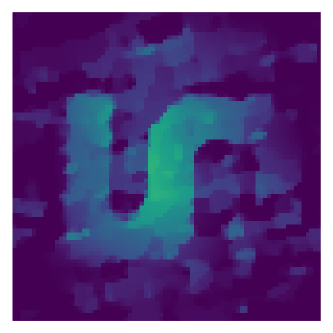}	}
	\subfloat[$\rho=1024$]{
\includegraphics{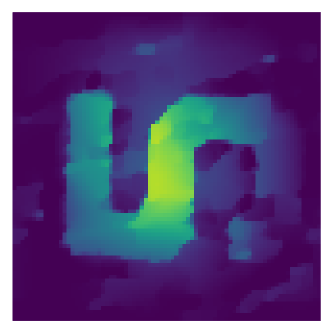}	}
	\subfloat[$\rho=1280$]{
\includegraphics{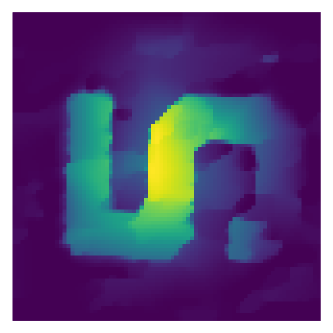}	}
	\\
	\subfloat{\raisebox{0.1\linewidth}{\small\texttt{\#2}}
	}	
	\subfloat[$M/N = 0.3$][$\rho = 655$\\$\bs{M/N=0.3}$]{
\includegraphics{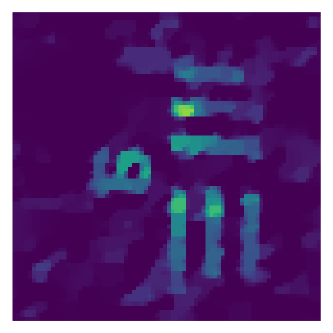}	}
	\subfloat[$M/N = 0.5$][$\rho = 1280$\\$\bs{M/N=0.5}$]{
\includegraphics{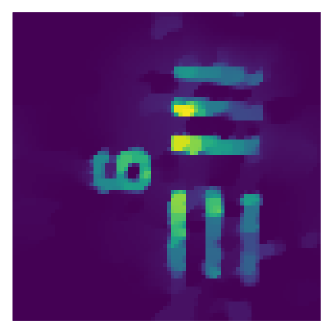}	}
	\subfloat[$M/N = 0.64$][$\rho = 1600$\\$\bs{M/N=0.64}$]{
\includegraphics{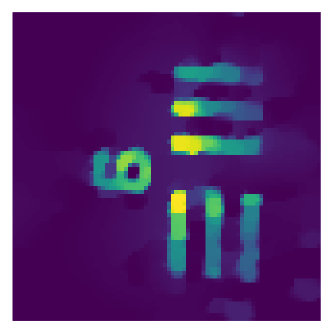}	}
	\caption{\label{fig:exp_results_USAF}\textbf{Vignetted estimates of \texttt{USAF} samples with PSS strategy.} All $80\times 80$ estimates are reconstructed from $M=4096$ observations (corresponding to $M/N = 0.64$) acquired with PSS strategy ($M_P=64$). Estimates for each sample share the same intensity scale. 
	}
      \end{figure}
\begin{figure}
	\centering
	\subfloat{\raisebox{0.1\linewidth}{\rotatebox[origin=c]{90}{\small$M/N=0.16$}}
	}
	\subfloat[Beads 1][\texttt{\#1} \\ $\rho = 440$]{
\includegraphics{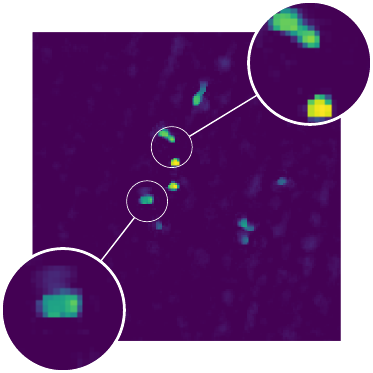}	}
	\subfloat[Beads 2][\texttt{\#2} \\ $\rho = 440$]{
\includegraphics{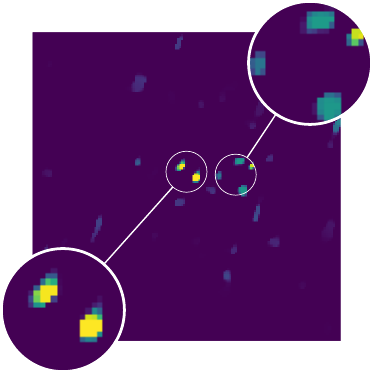}	}
	\subfloat[Beads 4][\texttt{\#4} \\ $\rho = 100$]{
\includegraphics{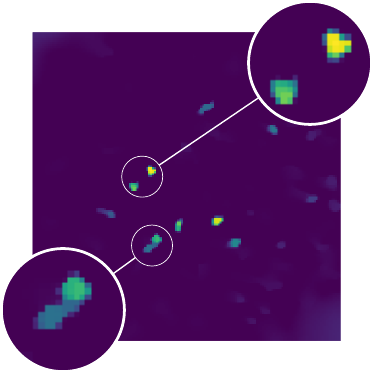}	}
	\caption{\label{fig:exp_results_beads}\textbf{Vignetted estimates of \texttt{Beads} samples (\texttt{\#1}, \texttt{\#2} and \texttt{\#4}).} All $160\times 160$ estimates are reconstructed from $M=4096$ observations acquired with PSS strategy ($M_P=64$). It corresponds to $M/N = 0.16$.}
\end{figure}

Reconstruction results are visible in \cref{fig:exp_results_USAF,fig:exp_results_beads_U-CS,fig:exp_results_beads}. Regarding \cref{fig:exp_results_beads_U-CS}, as expected, the SI strategy provides a better estimate for \texttt{Beads \#3} sample: beads are better resolved and there are less artifacts in the background. However, this quality is obtained at the cost of an acquisition time 60 times longer compared to the PSS strategy. Middle and right images are again obtained with two different strategies (PSS and SI) but the same acquisition time $t_{\text{acq}} = 6.8\,\text{s}$. In this case, the SI strategy fails to reconstruct the beads. 

For \texttt{USAF} samples and \texttt{Beads \#1}, \texttt{\#2} and \texttt{\#4}, we performed the image reconstruction from the observations acquired with PSS strategy for $M/N\in\{0.3,0.5,0.64\}$ (see \cref{fig:exp_results_USAF,fig:exp_results_beads}). We note that in order to minimize photo-bleaching, we keep the incident laser power extremely low (few 100\,s of $\mu$W over the entire FOV). This precludes us from acquiring  ground truth fluorescence images of the beads since the sensitivity of a standard camera is much lower than the single pixel detector.

\section{Conclusion}
\label{sec:perspectives}
We have proposed new sensing methods to image fluorescent object in the context of lensless endoscopy. Our procedure departs from the classical raster scanning imaging by using two constructions leveraging speckle patterns to illuminate the sample, namely the unstructured speckle imaging (SI) and the partial speckle scanning (PSS). Our work also relies on proving that these speckles correspond to random fields following a sub-exponential distribution at each location of the sample plane (see \cref{sec:speckle-distrib}). After normalization and discretization, they are thus good candidates to build efficient sensing matrices, such as sub-exponential random matrices~\cite{Adamczak2011}. 

To use speckle illumination to collect observations, several challenges must be solved with regard to the design of the fiber, the acquisition strategy, as well as the reconstruction scheme.  

First, the arrangement of the single mode cores must be optimized to achieve narrow autocorrelation of the speckle field (see \cref{sec:speckl-illum}), \ie a grain size as small as possible, and to minimize the magnitude of the side lobes. Fermat's golden spiral shape shows a low side lobes level and a very good contrast between the intensities of the central peak and the side lobes \cite{Sivankutty2018}. A small grain size combined with an appropriate choice of the pixel pitch leads to a sensing matrix with nearly independent columns, a desirable property in CS.

Second, the acquisition strategy must be thought to minimize the acquisition time to avoid as far as possible the loss of fluorescence of the sample and to reach a frame rate suitable for \emph{in vivo} imaging of cellular processes (\eg the propagation of nerve impulses). The SI strategy offers good reconstruction quality for far fewer measurements compared to RS (see \cref{fig:noise40_UCS_comparison}). However, the corresponding acquisition time is unrealistic for real biological applications, as requiring to change speckle patterns using slow SLM (with a frame rate of the order of $10\,{\rm ms}$). We overcome this limitation using the PSS strategy exploiting the memory effect of the fiber via the use of scan mirrors. Their rate is around 1,000 times higher compared to a change of the SLM configuration. With an appropriate value for the speckle shift $\delta$ (see \cref{fig:SCS_dmin}), we reach a reconstruction quality similar to that of SI (see \cref{fig:WSNR_tacq_level_curves_d2}). 

Finally, the reconstruction scheme, \ie the formulation and solving of the inverse problem, must encode prior information on the structure of the fluorophore density map. We proposed a debiased formulation including the minimization of TV or $\TGV$-norm. We showed that the sensing matrix can be approximated by a block-circulant matrix to reduce the number of stored patterns. When $M_P \le 16$, the reconstruction quality is still close to SI (see \cref{fig:WSNR_tacq_level_curves_d2-approx}).

As a perspective, let us mention that we only considered line scanning for the PSS strategy. However, if we use two scan mirrors, other scanning trajectories could be considered to further minimize the number of measurements compared to traditional RS. If the mirrors are driven parallely, there would be no increase in the acquisition time. In this case, we could perform speckle scanning in a way similar to RS: a single speckle would  scan $M\le N$ positions in the field of view (FOV). In addition, recent advances in ultrafast scanners employing acousto-optic deflectors~\cite{reddy2008three} and resonant galvanometers~\cite{fan1999video,Andresen2013} will serve to speed up the acquisitions both in the case of focused or speckle illumination. Faster spatial light modulators such as digital micromirror devices~\cite{geng2017digital} and deformable mirrors~\cite{bifano1999microelectromechanical} can also provide a massive speedup to SI, albeit at the cost of experimental integration.

The 2-D setup presented in this work is quite unrealistic for real \emph{in vivo} imaging: (\emph{i}) in practice, the LE will be required to image 3-D volumes and (\emph{ii}) we will not have access to the fiber endface to image the speckle pattern. Regarding those challenges, we would like to highlight two interesting research directions. The first one is the design of a 3-D acquisition and reconstruction framework. One of the challenges will be to deal with the depth dependency of the speckle autocorrelation \cite{Pascucci2017}. The second one is the study of blind imaging taking advantage of the memory effect like in \cite{Bertolotti2012,Stasio2016}. In this case, the estimation is based on the autocorrelations of the measurements vector and of the discretized speckle field (known \emph{a priori}). Another solution to deal with the inaccessibility to the fiber endface would be to perform a single initial calibration of the field and then, compute it. Given the high resilience of the fiber to bending, this would alleviate the tedious calibration of mulitple intensity patterns. This would provide a route to the deployment of LEs in a more realistic medical or diagnostic environment. 

\appendix
\section{Total variation and second-order total generalized variation}
\label{app:TGV}
The total variation (TV) and the total generalized variation (TGV) of an image are both quantities related to the first- or higher-order derivatives of this image. We start by defining the discrete \emph{gradient} operator $\nabla$ and the \emph{symmetrized derivative} operator $\varepsilon$.

\begin{definition}[from \cite{Guerit2015a}]
The general discrete gradient operator $\nabla$ applicable to $N\times k$ tensor fields is defined as
\begin{equation*}
	\nabla : \R{N\times k} \to \R{N\times 2k},~ \bs u \mapsto (\nabla \bs u) \coloneqq (\nabla_1 \bs u, \nabla_2 \bs u),
\end{equation*}
with $k\in\mathbb{N}_0$ and $\nabla_i \in \R{N\times N}$ the first spatial derivative of the tensor field in direction $\boldsymbol e_i$, aligned with axis $i$ of the 2-D image.
\end{definition}
				
\begin{definition}[from \cite{Guerit2015a}]
The symmetrized derivative operator $\varepsilon$ is 
\begin{equation*}
	\varepsilon : \R{N\times 2} \to \R{N\times 4},~ \bs u \mapsto \varepsilon(\bs u) \coloneqq \frac{1}{2}\left((\nabla\bs u) + (\nabla \bs u) S_{23}\right),
\end{equation*}
where $S_{23} \in \{0,1\}^{4\times 4}$ is a matrix permuting the second and the third columns of $(\nabla \bs u)$. 
\end{definition}
Applying $\varepsilon$ on the gradient of an image provides information about its second derivative. The concept of symmetrized tensors is detailed in \cite{Bredies2010}.

\begin{definition}
Let $\bs u_i \in\R{k}$ be the $i^{\text{th}}$ row of $\bs u$. The $L_{p,q}$-norm of $\bs u \in \R{N\times k}$ is defined as 
	$$\|\bs u\|_{p,q} = \left(\sum_{i=1}^N\|\bs u_i\|_p^q \right)^{\frac{1}{q}}.$$
\end{definition}

Regarding the previous definitions, the TV-norm of $\bs u \in\R{N}$ is defined in a discrete setting as the $\ell_1$-norm of the gradient magnitude of $\bs u$,
\begin{equation}
	\text{TV} : \R{N} \to \R{},~ \bs u \mapsto\text{TV}(\bs u) \coloneqq \|\nabla \bs u\|_{2,1}.
	\label{eq:TV_definition}
\end{equation}
Minimizing the TV-norm of an image in an estimation problem like \cref{eq:dist_regularized_tilde_TGV_CV_debiased} leads to piecewise constant estimate. If the original image is not piecewise constant, staircasing artifacts will appear \cite{Bredies2010}. To alleviate those limitations when working with real, piecewise smooth, images, we can resort to the TGV-norm. This norm was introduced by Bredies et al. \cite{Bredies2010} in 2010 and can be considered as the generalization of TV to higher-order image derivatives. The second-order $\TGV$-norm is defined as 
\begin{equation}
	\TGV : \R{N} \to \R{},~ \bs u \mapsto\TGV(\bs u) \coloneqq \min_{\bs w\in \R{N\times 2}}~\|\nabla \bs u - \bs w \|_{2,1} + \alpha \|\varepsilon(\bs w)\|_{2,1},
	\label{eq:TGV_definition}
\end{equation}
	where $\alpha > 0$ is a parameter making a trade-off between the edge-preserving term and the smoothness-promoting term. 
	Deriving $\TGV(\bs u)$ is not as easy as TV$(\bs u)$ because an additional minimization problem has to be solved. We re-write \cref{eq:dist_regularized_tilde_TGV_CV_debiased} in the following way,
\begin{equation*}
	\bs{\hat z}_\rho \in \underset{\bs z}{\text{arg~min}}~\|\bs R_\text{est}(\sqrt{M}\tilde{\bs \Phi} \bar{\bs S}\bs R_{\bs u}\bs z - \tilde{\bs y})\|^2 + \rho \|(\nabla \bs R_{\bs u} - \bs R_{\bs w}) \bs z \|_{2,1} + \rho \alpha\|\varepsilon(\bs R_{\bs w}\bs z)\|_{2,1}  + \imath_{\R{N}_+} (\bs R_{\bs u}\bs z),
\end{equation*}	
where $\bs z = (\bs u, \bs w)$, $\bs R_{\bs u}$ and $\bs R_{\bs w}$  are restriction operators keeping only the first column of $\bs z$ and the last two columns of $\bs z$, respectively. $\bs{\hat f}_\rho$ is given by $\bs R_{\bs u} \bs{\hat z}_\rho$.
\section*{Funding}

Part of this research is funded by the FNRS under Grant n$^\circ$ T.0136.20 (project Learn2Sense).

Projects related to lensless endoscopy in Fresnel Institute are supported by Agence Nationale de la Recherche FR NAIMA and National Institutes of Health US R21EB025389 \& R21MH117786.

\end{document}